\documentclass[preprint]{elsarticle}
\usepackage{xspace,comment}
\usepackage{amssymb,amscd}
\usepackage{amsfonts,stmaryrd}
\usepackage{amsmath}
\usepackage[utf8]{inputenc}
\usepackage{graphicx}
\usepackage{pifont}
\usepackage{pst-all}
\usepackage{tabls,subfig}
\usepackage{tikz}
\usepackage{enumerate}
\newcommand{\z}{\ensuremath{\mathbb{Z}}\xspace}

\newcommand{\n}{\ensuremath{\mathbb{N}}\xspace}

\newdefinition{definition}{Definition}[section]
\newdefinition{remark}{\normalfont \it Remark}
\newdefinition{example}{\normalfont \it Example}
\newtheorem{theorem}{Theorem}[section]
\newtheorem{proposition}[theorem]{Proposition}
\newtheorem{lemma}[theorem]{Lemma}

\newproof{proof}{\textit{Proof}}

\newcommand{\ie}{i.e.\@\xspace}
\newcommand{\an}{\ensuremath{A^{\mathbb{N}}}\xspace}
\newcommand{\az}{\ensuremath{A^{\mathbb{Z}}}\xspace}
\newcommand{\bz}{\ensuremath{B^{\mathbb{Z}}}\xspace}

\newcommand{\set}[1]{\left\{#1\right\}}

\newcommand{\para}[1]{(#1)}

\newcommand{\wrt}{w.r.t.\@\xspace}

\newcommand{\ignore}[1]{}

\newcommand{\CA}{\mathcal{CA}}
\newcommand{\nca}{\ensuremath{\nu\text{-}\mathcal{CA}}\xspace}

\renewcommand{\H}{\ensuremath{H^{(\theexample)}}\xspace}

\begin{document}
\begin{frontmatter}
\title{Non--Uniform Cellular Automata:\\ classes, dynamics, and decidability\tnoteref{blabla}}
\tnotetext[blabla]{This is an extended and improved version of the paper~\cite{cattaneo09}
presented at LATA2009 conference.}
\author[Nice]{Alberto Dennunzio\corref{cor}}
\ead{alberto.dennunzio@unice.fr}

\author[Nice]{Enrico Formenti\corref{cor}}
\ead{enrico.formenti@unice.fr}

\author[Nice]{Julien Provillard}
\ead{julien.provillard@i3s.unice.fr}

\cortext[cor]{Corresponding author.}


\address[Nice]{Universit\'e Nice-Sophia Antipolis,
Laboratoire I3S, 2000 Route des Colles, 06903 Sophia Antipolis
(France)}

\begin{abstract}
The dynamical behavior of non-uniform cellular automata is compared with the one of classical cellular automata. Several differences and
similarities are pointed out by a series of examples. Decidability of
basic properties like surjectivity and injectivity is also established. The final part studies a strong form of equicontinuity property specially
suited for non-uniform cellular automata. 
\end{abstract}
\begin{keyword}
cellular automata \sep non--uniform cellular automata \sep
decidability \sep symbolic dynamics
\end{keyword}
\end{frontmatter}

\section{Introduction}
A complex system is (roughly) defined  by
a multitude of simple individuals which cooperate to build a complex (unexpected) global behavior by local interactions.
Cellular automata (CA) are often used to model complex systems when individuals are embedded in a uniform ``universe'' in which local interactions are
the same for all. Indeed, a cellular automaton is made of identical
finite automata arranged on a regular lattice. Each automaton updates its state by a local rule on the basis of its state and the one of a fixed set of neighbors.
At each time-step, the same (here comes uniformity) local rule is
applied to all finite automata in the lattice. For recent results
on CA dynamics and an up-to-date bibliography see for
instance~\cite{formenti07b,kurka07,CDF08,kari08,
dilena10, Dennunzio08ca,capka09,
Dennunzio09jb,Dennunzio09jc,Dennunzio09jd}.

In a number of situations one needs a more general setting. 
One possibility consists in relaxing the uniformity constraint.
This choice may result winning for example for

\paragraph{Complexity design control} In many phenomena, each individual locally interacts with others but maybe these interactions depend on the individual itself or on its position in the space. For example, when studying the formation of hyper-structures in cells, proteins move in the cellular soup and do not behave just like 
billiard balls. They chemically interact each other only when they meet special situations (see for instance, \cite{kier05}) or when they are at some special places (in rybosomes for instance). It is clear that one might simulate all those situations by a CA but the writing of a single local rule will be an excessive difficult task, difficult to control. A better option would be to write simpler (but different) local rules that are applied only at precise positions so that less constraints are to be taken into account at each time.
 
\paragraph{Structural stability} Assume that we are investigating the robustness of a system \wrt some specific property $P$. If some individuals change their ``standard'' behavior does the system still have property $P$? What is the ``largest'' number of individuals that can change their default behavior so that the system does not change its overall evolution?

\paragraph{Reliability} CA are more and more used to perform fast parallel computations (beginning from \cite{chaudhuri97}, for example). Each cell of the CA can be implemented by a
simple electronic device (FPGAs for example) \cite{nakano02}. Then,
how reliable are computations w.r.t. failure of some of these
devices? (Here failure is interpreted as a device which behaves
differently from its ``default'' way).
\medskip

The generalization of CA to non-uniform CA (\nca) has some interest
in its own since the new model coincides with the set of continuous functions in Cantor topology. It is clear that the class of continuous
functions is too large to be studied fruitfully. In the present paper,
we present several sub-classes that are also interesting in applications. First of all, we show that several classical results
about the dynamical behavior of CA are no longer valid in the new setting.
Even when the analysis is restricted to smaller classes of non-uniform CA, the overall impression is that new stronger techniques will be
necessary to study \nca. However, by generalizing the notion of
De Bruijn graph, we could prove the decidability of basic set properties
like surjectivity and injectivity. We recall that these property are often necessary conditions of many classical definitions of deterministic chaos.

Keeping on with surjectivity and injectivity,
we give a partial answer about reliability and structural stability questions issued above. More precisely,
we answer the following question: assuming to perturb some CA in some finite number of sites, if one knows that the corresponding \nca is surjective (resp., injective) does this imply that the original CA was surjective (resp., injective)?

The final part starts going more in deep with the study of the long-term
dynamical behavior of \nca. Indeed, we show that under some conditions, if a \nca is a perturbed version of some equicontinuous or almost equicontinuous CA, then it shares the same dynamics.

Finally, we develop some complex examples showing that even small perturbations of an almost equicontinuous CA can give raise to
sensitive to initial conditions behavior or to equicontinuous dynamics.

We might conclude that breaking uniformity property in a CA may cause a dramatic change in the dynamical behavior.

\section{Background}
In this section, we briefly recall standard definitions about CA
and dynamical systems. For introductory matter 
see~\cite{kurka04}
, for instance.
%
%
For all $i,j\in\z$ with $i\leq j$ (resp., $i<j$), let
$[i,j]=\set{i,i+1,\ldots,j}$ (resp., $[i,j)=\set{i,\ldots,j-1}$). Let $\n_+$ be the set of positive
integers.
\paragraph{\textbf{Configurations and CA}}
Let $A$ be an alphabet. A \emph{configuration} is
a function from $\z$ to $A$. The \emph{configuration set} $\az$ is
usually equipped with the metric $d$ defined as follows
\[
\forall x,y\in\az,\;
d(x,y)=2^{-n},\;\text{where}\;n=\min\set{i\geq 0\,:\,x_i\ne y_i
\;\text{or}\;x_{-i}\ne y_{-i}}\enspace.
\]
When $A$ is finite, the set $\az$ is a compact, totally
disconnected and perfect topological space (\ie, $\az$ is a Cantor
space). For any pair $i,j\in\z$, with $i\leq j$, and any
configuration $x\in\az$ we denote by $x_{[i,j]}$ the word
$x_i\cdots x_j\in A^{j-i+1}$, \ie, the portion of $x$ inside the
interval $[i,j]$. Similarly, $u_{[i,j]}=u_i\cdots u_j$ is the portion of a word $u\in A^l$ inside $[i,j]$ (here, $i,j\in [0,l)$). In both the previous notations, $[i,j]$ can be replaced by $[i,j)$ with the obvious meaning. For any word $u\in A^*$, $|u|$ denotes its length. A \emph{cylinder} of block $u\in A^k$ and
position $i\in\z$ is the set $[u]_i=\{x\in A^{\z}: x_{[i,
i+k)}=u\}$. Cylinders are clopen sets \wrt the metric $d$ and
they form a basis for the topology induced by $d$. A configuration $x$ is said to be \emph{$a$-finite} for some $a\in A$ if there exists $k\in\n$ such that $x_i=a$ for any $i\notin [-k,k]$. In the sequel, the collection of the $a$-finite configurations for a certain $a$ will be simply called set of finite configurations.

A (one--dimensional) \emph{CA} is a structure $\langle A, r,
f\rangle$, where $A$ is the \emph{alphabet}, $r\in\n$ is the
\emph{radius} and $f: A^{2r+1} \to A$ is the \emph{local rule} of
the automaton. The local rule $f$ induces a \emph{global rule}
$F:\az\to\az$ defined as follows,
\begin{equation}\label{eq:caglobalrule}
\forall x\in \az,\,\forall i\in\z ,\quad F(x)_i= f(x_{i-r},
\ldots, x_{i+r})\enspace.
\end{equation}
Recall that $F$ is a uniformly continuous map \wrt the metric $d$. 

With an abuse of notation, a CA local rule $f$ is extended to the function $f:A^*\to A^*$ which map any $u\in A^*$ of length $l$ to the word $f(u)$ such that $f(u)=\epsilon$ (empty word), if $l\leq 2r$, and $f(u)_i=f(u_{[i,i+2r]})$ for each $i\in[0,l-2r)$, otherwise.
\ignore{
A CA with global rule $F$ is \emph{right} (resp., \emph{left})
\emph{closing} iff $F(x)\neq F(y)$ for any pair $x,y\in\az$ of
distinct left (resp., right) asymptotic configurations, \ie,
$x_{(-\infty,n]}=y_{(-\infty,n]}$ (resp.,
$x_{[n,\infty)}=y_{[n,\infty)}$) for some $n\in\z$, where
$a_{(-\infty,n]}$ (resp., $a_{[n, \infty)}$) denotes the portion
of a configuration $a$ inside the infinite integer interval
$(-\infty,n]$ (resp., $[n, \infty)$). A CA is said to be
\emph{closing} if it is either left or right closing. A rule
$f:A^{2r+1} \to A$ is \emph{rightmost} (resp., \emph{leftmost})
\emph{permutive} iff $\forall u\in A^{2r}, \forall\beta\in
A,\exists \alpha\in A$ such that $f(u\alpha)=\beta$ (resp.,
$f(\alpha u)=\beta$). A 1D CA is said to be \emph{permutive} if
its local rule is either rightmost or leftmost permutive.
}

\paragraph{\textbf{DTDS}} A \emph{discrete time dynamical system (DTDS)}
is a pair $\para{X,G}$ where $X$ is a set equipped with a distance
$d$ and $G:X\mapsto X$ is a map which is continuous on $X$ with
respect to the metric $d$. When $\az$ is the configuration space
equipped with the above introduced metric and $F$ is the global
rule of a CA, the pair $\para{\az,F}$ is a DTDS. From now on, for
the sake of simplicity, we identify a CA with the dynamical system
induced by itself or even with its global rule $F$.

Given a DTDS $\para{X,G}$, an element $x\in X$ is an
\emph{ultimately periodic point} if there exist $p,q\in\n$ such
that $G^{p+q}(x)=G^q(x)$. If $q=0$, then $x$ is a \emph{periodic
point}, \ie, $G^p(x)=x$. The minimum $p$ for which $G^p(x)=x$
holds is called \emph{period} of $x$. 
\ignore{
If the set of all periodic
points of $G$ is dense in $X$, we say that the DTDS has
\emph{dense periodic orbits (DPO)}.
}
%
%

Recall that a DTDS $\para{X,g}$ is \emph{sensitive to the initial
conditions} (or simply \emph{sensitive}) if there exists a
constant $\varepsilon>0$ such that for any element $x\in X$ and
any $\delta>0$ there is a point $y\in X$ such that $d(y,x)<\delta$
and $d(G^n(y),G^n(x))>\varepsilon$ for some $n \in\n$. A DTDS
$\para{X,G}$ is \emph{positively expansive} if there exists a
constant $\varepsilon>0$ such that for any pair of distinct
elements $x,y$ we have $d(G^n(y),G^n(x))\geq\varepsilon$ for some
$n\in\n$. If $X$ is a perfect set, positive expansivity implies
sensitivity.
Recall that a DTDS $\para{X,g}$ is \emph{(topologically)
transitive} if for any pair of non-empty open sets
$U,V\subseteq X$ there exists an integer $n\in\n$ such that
$g^n(U)\cap V\ne\emptyset$.
\ignore{
A DTDS $\para{X,G}$ is \emph{(topologically) mixing} if for any
pair of non-empty open sets $U,V\subseteq X$ there exists an
integer $n\in\n$ such that for any $t\geq n$ we have $G^t(U)\cap
V\ne\emptyset$.

A DTDS $\para{X,g}$ is \emph{(topologically) strongly transitive}
if for any  non-empty open set $U\subseteq X$, it holds that
$\bigcup_{n\in\n} G^n(U)=X$. 
}
A DTDS $\para{X,G}$ is \emph{sujective}
(resp., \emph{injective}) iff  $G$ is surjective (resp., $G$ is
injective).


Let $\para{X,G}$ be a DTDS. An element $x\in X$ is
 an \emph{equicontinuity point} for $G$ if $\forall\varepsilon>0$
there exists $\delta>0$ such that for all $y\in X$,
$d(y,x)<\delta$ implies that $\forall
n\in\n,\;d(G^n(y),G^n(x))<\varepsilon$. For a CA $F$, the
existence of an equicontinuity point is related to the existence
of a special word, called \emph{blocking word}. A word $u\in A^k$
is $s$-blocking ($s\leq k$) for a CA $F$ if there exists an offset
$j\in [0, k-s]$ such that for any $x,y\in [u]_0$ and any
$n\in\n$, $F(x)_{[j,j+s-1]}=F(y)_{[j,j+s-1]}$\,. A word $u\in A^k$
is said to be \emph{blocking} if it is $s$-blocking for some
$s\leq k$. A DTDS is said to be \emph{equicontinuous} if
$\forall\varepsilon>0$ there exists $\delta>0$ such that for all
$x,y\in X$, $d(x,y)<\delta$ implies that $\forall
n\in\n,\;d(G^n(x),G^n(y))<\varepsilon$. If $X$ is a compact set, a
DTDS $\para{X,G}$ is equicontinuous iff the set $E$ of all its
equicontinuity points is the whole $X$. A DTDS is said to be
\emph{almost equicontinuous} if $E$ is residual (\ie, $E$ contains
an infinite intersection of dense open subsets). 
In~\cite{kurka97}, K\r{u}rka proved that a CA is almost
equicontinuous iff it is non-sensitive iff it admits a
$r$-blocking word.

Recall that two DTDS $\para{X,G}$ and $\para{X^\prime,G^{\prime}}$
are \emph{topologically conjugated} if
there exists a homeomorphism $\phi:X\mapsto
X^\prime$ such that
$G^\prime\circ \phi=\phi\circ G$. 
In that case, $\para{X,G}$ and $\para{X^\prime,G^{\prime}}$ share some properties such as surjectivity, injectivity, transitivity.
\section{Non--Uniform Cellular Automata}
The meaning of~\eqref{eq:caglobalrule} is that the same local rule
$f$ is applied to each site of the CA. Relaxing this 
constraint gives us the definition of a \nca. More formally one
can give the following notion.

\begin{definition}[Non--Uniform Cellular Automaton (\nca)]\mbox{}\\
A \emph{Non--Uniform Cellular Automaton} (\nca) is a structure
$\para{A, \{h_i,r_i\}_{i\in\z}}$ defined by a family of local
rules $h_i: A^{2r_i+1} \to A$ of radius $r_i$ all based on the
same alphabet $A$.
\end{definition}
Similarly to CA, one can define the global rule of a \nca as the
map $H:\az\to\az$ given by the law
\[
\forall x\in\az,\, \forall i\in\z,\quad H(x)_i =
h_i(x_{i-r_i},\dots,x_{i+r_i})\enspace.
\]
From now on, we identify a \nca (resp., CA) with the discrete
dynamical system induced by itself or even with its global rule
$H$ (resp., $F$).

It is well known that the Hedlund's Theorem~\cite{hedlund69}
characterizes CA as the class of continuous functions commuting
with the shift map $\sigma : A^\z\to A^\z$, where $\forall
x\in\az,\forall i\in\z,\sigma(x)_i=x_{i+1}$. 
It is
straightforward to prove that
a function $H:\az\to\az$ is the global map of a \nca iff it is
continuous.
\ie, in other words, 
iff the pair $(\az,H)$ is a DTDS.
Remark that the definition of \nca is by far too general to be useful.
Therefore, we are going to focus our attention only over three
special subclasses of \nca.

\begin{definition}[d\nca]
A \nca $\para{A, \{h_i,r_i\}_{i\in\z}}$ is a d\nca if there exist two naturals $k,r$ and a rule
$h: A^{2r+1}\to A$ such that $h_i=h$ for all integers $i$ with $|i|>k$. In this case, we say that the given \nca has $h$ as \emph{default rule}.
\end{definition}
\begin{definition}[p\nca]
A \nca $\para{A, \{h_i,r_i\}_{i\in\z}}$ is a p\nca if there exist two naturals $k$,$r$, a \emph{structural period} $p>0$, and two sets $\{f_0,\ldots, f_{p-1}\}$ and $\{g_0,\ldots, g_{p-1}\}$ of rules of radius $r$ such that for any integer $i$ with $|i|>k$
\[
h_i=
\begin{cases}
f_{i\,\text{mod}\, p} & \text{if $i>k$}\\
g_{i\,\text{mod}\, p} & \text{if $i<-k$}
\end{cases}
\]
If $p=1$, we say that the given \nca has $f_0$ and $g_0$ as right and left default rules, respectively.
\end{definition}
\begin{definition}[r\nca]
A \nca $H=\para{A, \{h_i,r_i\}_{i\in\z}}$ is a r\nca if there exists an integer $r$ such that $\forall i\in\z$, $r_i=r$. In this case, we say that $H$ has radius $r$.
\end{definition}
The first two class restrict the number of positions at which
non-default rules can appear, while the third class restricts the
number of different rules but not the number of occurrences nor it
imposes the presence of default rules. Some simple examples follow.

\begin{example}\label{ex:ncanotca}
Consider the \nca $\H:\az\to\az$ defined as
\[
\forall x\in\az,\quad \H(x)_i =
\begin{cases}
1 &\text{if $i=0$} \\
0 &\text{otherwise}
\end{cases}\enspace.
\]
Remark that \H is a d\nca which cannot be a CA since it does not
commute with $\sigma$. This trivially shows that the class of \nca is larger than the one of CA.
\end{example}
\begin{example}
\label{ex3}
Consider the \nca $\H:\az\to\az$ defined as
\[
\forall x\in\az,\quad \H(x)_i =
\begin{cases}
x_{i+1} &\text{if $i<0$} \\
x_{0} &\text{if $i=0$} \\
x_{i-1} &\text{if $i>0$} \\
\end{cases}
\]
Remark that \H is a p\nca (with $p=1$) but not a d\nca.
\end{example}
\begin{example}
Consider the \nca $\H:\az\to\az$ defined as
\[
\forall x\in\az,\quad \H(x)_i =
\begin{cases}
1 &\text{if $i$ is even} \\
0 &\text{otherwise.}
\end{cases}
\]
Remark that \H is a r\nca but not a p\nca.
\end{example}
\begin{example}\label{ex:ncanotrnca}
Consider the \nca $\H:\az\to\az$ defined as
$\forall x \in\az, \H(x)_i = x_0$.
Remark that \H is a \nca but not a r\nca.
\end{example}

We give some relationships and properties involving the classes of \nca above introduced. 
\begin{proposition}
$\CA\varsubsetneq d\nca\varsubsetneq p\nca\varsubsetneq r\nca\varsubsetneq\nca$,
where $\CA$ is the set of all CA.
\end{proposition}
\begin{proof}
The inclusions $\subseteq$ easily follow from the definitions. For
the strict inclusions refer to Examples~\ref{ex:ncanotca} to
\ref{ex:ncanotrnca}.\qed
\end{proof}

Similarly to what happens in the context of CA one can prove the following.

\begin{proposition}
\label{prop:conjug}
Any r\nca is topologically conjugated to a r\nca of radius 1.
\end{proposition}
\begin{proof}
Let $H$ be a r\nca on the alphabet $A$.
If $H$ has radius $r=1$ then the statement is trivially true. 
Otherwise, let $B = A^r$ and
define $\phi : A^{\mathbb{Z}} \rightarrow B^{\mathbb{Z}}$ as
$\forall i \in \mathbb{Z}, \phi(x)_i = x_{[ir,(i+1)r)}$.
Then, it is not difficult to see that the r\nca
$(B^{\mathbb{Z}},H')$ of radius $1$ defined as $ \forall x \in
A^{\mathbb{Z}}, \forall i \in \mathbb{Z}, H'(x)_i =
h'_i(x_{i-1},x_i,x_{i+1})$
 is topologically conjugated to $H$ via $\phi$,
where $\forall u,v,w \in B, \forall i \in \mathbb{Z}, \forall j
\in \{0,\dots,r-1\}, (h'_i(u,v,w))_j =
h_{ir+j}(u_{[j,r)}vw_{[0,j]})$.\qed
\end{proof}

Finally, the following result shows that every r\nca is a
subsystem of a suitable CA. 
\begin{theorem}
Any r\nca $H:\az\to\az$ of radius $r$ is a subsystem of a  CA,
\ie, there exist a CA $F:\bz\to\bz$ on a suitable alphabet $B$ and
a continuous injection $\phi:\az\to\bz$ such that
$\phi\circ H=F\circ\phi$.
\end{theorem}
\ignore{
\[
\begin{diagram}
\node{A^{\mathbb{Z}}} \arrow{e,t,->}{H} \arrow{s,l,->}{\phi} \node[1]{A^{\mathbb{Z}}} \arrow{s,r,->}{\phi} \\
\node{B^{\mathbb{Z}}} \arrow{e,b,->}{F} \node[1]{B^{\mathbb{Z}}}
\end{diagram}
\]
}
\begin{proof}
Consider a r\nca $H : A^{\mathbb{Z}} \rightarrow A^{\mathbb{Z}}$
of radius $r$. Remark that there are only $n=|A|^{|A|^{2r+1}}$
distinct functions $h_i : A^{2r+1} \rightarrow A$. Take a
numbering $(f_j)_{1 \leq j\leq n}$ of these functions and let $B =
A \times \{1, \dots, n\}$. Define the map $\phi :
A^{\mathbb{Z}} \rightarrow B^{\mathbb{Z}}$ such that $\forall x
\in A^{\mathbb{Z}}, \forall i \in \mathbb{Z}, \phi(x)_i =
(x_i,k)$, where $k$ is the integer for which $H(x)_i =
f_k(x_{i-r},\dots,x_{i+r})$. Clearly, $\phi$ is injective and
continuous. Now, define a CA $F : B^{\mathbb{Z}} \rightarrow
B^{\mathbb{Z}}$ using the local rule $f : B^{2r+1} \rightarrow B$
such that
\[
f((x_{-r},k_{-r}),\dots,(x_0,k_0),\dots,(x_r,k_r)) =
(f_{k_0}(x_{-r},\dots,x_r),k_0) \enspace.
\]
It is not difficult to see that $\phi \circ H = F\circ \phi$.\qed
\end{proof}
\section{CA versus \nca}
\label{inj_surj}
In this section, we illustrate some differences in dynamical
behavior between CA and \nca. The following properties which are
really specific for CA are lost in the larger class of \nca.
%
\begin{enumerate}
\item[P1)] \emph{the set of ultimately periodic points is dense in
$\az$}.
\item[P2)] \emph{surjectivity $\Leftrightarrow$ injectivity on
finite configurations.}
\item[P3)] \emph{surjectivity $\Leftrightarrow$ any configuration
has a finite number of pre--images}.
\item[P4)] \emph{expansivity $\Rightarrow$ transitivity}
\item[P5)] \emph{expansivity $\Rightarrow$ surjectivity}
\item[P6)] \emph{injectivity $\Rightarrow$ surjectivity}
\end{enumerate}
%
Some of the previous properties are not valid for the following
\nca.
\begin{example}
Let $A = \{0,1\}$ and define the d\nca $\H:\az\to\az$ as
\[
\forall x\in\az,\forall i\in\z,\quad \H(x)_i =
\begin{cases}
x_i &\text{if $i=0$} \\
x_{i-1} &\text{otherwise}\enspace.
\end{cases}
\]
%
\begin{description}
\item P1) is not valid for $\H$.
\end{description}
\begin{proof}
Let $H=\H$. Denote by $P$ and $U$ the sets of periodic and
ultimately periodic points, respectively. Let $E =\{x\in\az:
\forall i \in\mathbb{N}, x_i = x_0 \}$. Take $x\in P$ with $H^p(x)
= x$. Remark that the set $B_x = \{i\in\n: x_i \neq x_0 \}$ is
empty. Indeed, by contradiction, assume that $B\ne\emptyset$ and
let $m=\min B$. It is easy to check that $\forall y \in\az,
\forall i\in\n, H^i(y)_{[0,i]} = {y_0}^{i+1}$, hence $x_{m} =
H^{pm}(x)_{m} = x_0$, contradiction. Thus $x\in E$ and $P
\subseteq E$.

Let $y \in H^{-1}(E)$. We show that $B_y=\emptyset$. By
contradiction, let $n=\min B_y$. Since $H(y)_{n + 1} = y_{n} \neq
y_0 = H(y)_0$, then $H(y) \notin E$. Contradiction, then $y \in E$
and $H^{-1}(E) \subseteq E$. So $ \forall n \in \mathbb{N},
H^{-n}(E) \subseteq E$. Moreover,
$U= \bigcup_{n\in \n} H^{-n}(P) \subseteq \bigcup_{n \in
\mathbb{N}} H^{-n}(E) \subseteq E $
and $E$ is not dense.\qed
\end{proof}
\begin{description}
\item P3) is not valid for $\H$
\end{description}
\begin{proof}
We show that $\H$ has no configuration with an infinite
number of pre-images although it is not surjective. In particular,
any configuration has either $0$ or $2$ pre-images.

First of all, $\H$ is not surjective. Indeed, since $\forall
x \in\az, \H(x)_0 = \H(x)_1$, configurations in the set
$B =\{ x\in \az : x_0 \ne x_1\}$ have no pre-image. Furthermore,
any $x\in {\az} \setminus B$ has $y$ and $z$ as unique
pre--images, where $y$ and $z$ are configurations such that
$\forall i\notin \{-1,0\}, y_i = z_i = x_{i+1}, y_0 = z_0 = x_0,
y_{-1} = 0; z_{-1} = 1$.\qed
\end{proof}
We stress that $\H$ is not surjective, despite it is based on
two local rules each of which generates a surjective CA (namely,
the identity CA and the shift CA).
\end{example}
 In order to explore other properties, we introduce an other \nca.
\begin{example} 
Let $A = \{0,1\}$ and define $\H:\az\to\az$ by
\[
\forall x\in\az,\forall i\in\z,\quad \H(x)_i =
\begin{cases}
0 &\text{if $i=0$} \\
x_{i-1}\oplus x_{i+1} &\text{otherwise}\enspace,
\end{cases}
\]
where $\oplus$ is the xor operator.
\end{example}
\begin{description}
\item P2) is not valid for $\H$.
\end{description}
\begin{proof}
We prove that $\H$ is injective on the $0$-finite
configurations but it is not surjective. It is evident that
$\H$ is not surjective. Let $x,y$ be two finite
configurations such that $\H(x) = \H(y)$. By
contradiction, assume that $x_i \neq y_i$, for some $i\in\z$.
Without loss of generality, assume that $i\in\n$. Since $x_{i}
\oplus x_{i+2} = \H(x)_{i+1} = \H(y)_{i+1} = y_{i}
\oplus y_{i+2}$,
 it holds that $x_{i+2} \neq y_{i+2}$ and, by induction,
$\forall j \in\n, x_{i+2j} \neq y_{i+2j}$. We conclude that
$\forall j \in\n, x_{i+2j} = 1$ or $y_{i+2j} = 1$ contradicting
the assumption that $x$ and $y$ are finite.\qed
\end{proof}
\begin{description}
\item P4) and P5) are not valid for $\H$.
\end{description}
\begin{proof}
Let $H=\H$. $H$ is not transitive since it is not surjective.
We show that $H$ is positively expansive. Let $x$ and
$y$ be two distinct configurations and let $k=\min_{i\in\z}\{|i|, x_i
\neq y_i \}$. If $k=0$, then $d(H^0(x),H^0(y))=1\geq
\frac{1}{2}$. Without loss of generality assume $k=n>0$. 
It is clear that $H(x)_{n-1}=x_{n-2} \oplus x_n \neq
y_{n-2} \oplus y_n = H(y)_{n-1}$ and $H(x)_{[0,n-2]} =
H(y)_{[0,n-2]}$. Iterating the same reasoning one sees that
$H^{n-1}(x)_1 \neq H^{n-1}(y)_1$. Hence $d(H^{n-1}(x),H^{n-1}(y))
\geq \frac{1}{2}$. Thus $H$ is positively expansive with
expansivity constant $\frac{1}{2}$.\qed
\end{proof}

Consider now the \nca $H^{(2)}$ from Example~\ref{ex3}.
\begin{description}
\item P6) is not valid for $H^{(2)}$.
\end{description}
\begin{proof}
Concerning non-surjectivity, just remark that
only configurations $x$ such that $x_{-1} = x_0 = x_1$ have a
pre-image. Let $x,y\in\az$ with $H^{(2)}(x) = H^{(2)}(y)$. Then, we have
$\forall i > 0, x_{i-1} = y_{i-1}$ and $\forall i < 0, x_{i+1} =
y_{i+1}$. So $x=y$ and $H$ is injective.\qed
\end{proof}
%
\section{Basic Properties of \nca and Decidability}

This section is centered on two fundamental properties, namely surjectivity and injectivity.
Focusing on CA, both properties are strongly related to
peculiar dynamical behaviors. Injectivity coincides with reversibility~\cite{hedlund69}, while surjectivity is a necessary condition for almost
all the widest known definitions of deterministic chaos (see~\cite{cattaneo99}, for instance).
\smallskip

In (1D) CA settings, the notion of De Bruijn graph is very
handy to find fast decision algorithms for surjectivity,
injectivity and openness~\cite{sutner91}. Here, we extend this notion to work with p\nca having period 1 and find decision algorithm for surjectivity. We stress that
surjectivity is undecidable for two (or higher) dimensional p\nca,
since surjectivity is undecidable for 2D CA~\cite{kari94a}.

\begin{definition}\label{def:debrujin}
Consider a p\nca $H$ of radius $r$ and period $p=1$ having $f$ and $g$ as right and left default rules. Let $k\in\n$ be the largest natural such that $h_k\ne f$ or $h_{-k}\ne g$. The \emph{De Bruijn graph} of $H$ is the triple
$G=(V,E,\ell_G)$ where $V = A^{2r} \times \{-k,\dots,k+1\}$ and $E$ is the set of pairs $((u,\alpha),(v,\beta))\in V^2$ with label $\ell_G((u,\alpha),(v,\beta))$ in $A\times \{0,1\}$ such that $u_{[1,2r)}=v_{[0,2r-1)}$ and one of the following conditions is verified
\begin{enumerate}
\item[a)] $\alpha = \beta = -k$; in this case the label is $(g(u_0v),0)$
\item[b)] $\alpha+1 = \beta$; in this case the label is $(h_\alpha(u_0v),0)$
 \item[c)] $\alpha= \beta = k+1$; in this case the label is $(f(u_0v),1)$
\end{enumerate}
\end{definition}
By this graph, a configuration can be seen as a bi-infinite path
on vertexes which passes once from a vertex whose second component
is in $[-k+1,k]$ and infinite times through other vertices.  The
second component of vertices allows to single out the positions of
local rules different from the default one. The image of a
configuration is the sequence of first components of edge labels.


\begin{lemma}
\label{lem:surj-dec-p-1}
Surjectivity is decidable for p\nca with structural period $p=1$.
\end{lemma}

\begin{proof}
Let $H$ be a p\nca with structural period $p=1$ and let $G$ be its De Bruijn graph. We prove that $H$ is surjective iff  $G$ recognizes the language $(A \times \{0\})^{*}(A \times \{1\})^{*}$ when $G$ is considered as the graph of a finite state automaton.

Let $k$ be as in Definition~\ref{def:debrujin} and denote by $(w, s)$ any word of $(A \times \{0,1\})^*$ with  $w\in A^l$, $s\in\{0,1\}^l$, for some $l\in\n$. 

Assume that $H$ is surjective and take $(w',s) \in (A \times
\{0\})^{n}(A \times \{1\})^{*}$, for any $n\in\n$. Then, there exists $x\in\az$ such that $H(x)_{[m,m+l)}=w'$, where $m=k+1-n$. Set $w=x_{[m-r,m+l+r)}$. Hence, the word $(w',s)$ is the sequence of edge labels of the following vertex path on $G$: 
\[
(w_{[0,2r)},\alpha_0),\dots,(w_{[l,l+2r)},\alpha_{l})
\]
where
\[
\alpha_i = \left\{
\begin{array}{cl}
-k & \text{ if } m+i < - k \\
k+1 & \text{ if } m+i > k \\
m+i & \text{ otherwise}
\end{array}
\right.
\]
For the opposite implication, assume that $G$ recognizes $(A
\times \{0\})^*(A\times \{1\})^*$. Take $y\in\az$ and let $n
> k$. Since $G$ recognizes $(y_{[-n,n]},0^{n+k+1}1^{n-k})$, there exists
$x\in\az$ such that  $H(x)_{[-n,n]} = y_{[-n,n]}$. Set
$X_n=\set{x\in\az, x_{[n,n]} = y_{[-n,n]}}$. For any $n\in\n$,
$X_n$ is non-empty and compact. Moreover, $X_{n+1}\subseteq
X_{n}$. Therefore, $X=\bigcap_{n\in\n}X_n\ne\emptyset$ and
$H(X)=\set{y}$. Hence, $H$ is surjective.\qed
\end{proof}
In order to deal with injectivity, we introduce the following  notion.
%
\begin{definition}
Consider a p$\nca$ of structural period $1$ and let $G = (V,E,\ell_G)$
be its De Bruijn graph. The \emph{product graph} $P$ of $H$ is a
labeled graph $P = (V \times V, W,\ell_P)$ where $(((u,\alpha),(v,\beta)),((w,\gamma),(z,\delta))) \in W$ iff
\[
\begin{cases}
 \alpha=\beta\;\text{and}\;\gamma = \delta\\
 ((u,\alpha),(w,\gamma)) \in E\;\text{and}\;((v,\beta),(z,\delta)) \in E\\
 \ell_G((u,\alpha),(w,\gamma))=\ell_G((v,\beta),(z,\delta))
\end{cases}
\]
and $\ell_P: W\to A$ is defined as
\[
\ell_P((((u,\alpha),(v,\beta)),((w,\gamma),(z,\delta))))=
\left\{
\begin{array}{ll}
0,&\text{if}\;u=v\;\text{and}\;w=z\\
1,&\text{otherwise.}
\end{array}
\right.
\]
The \emph{reduced product graph} $D$ of $P$ 
is the sub-graph of $P$ made by all the strongly connected
components of $P$ plus all nodes and edges belonging to any path between two connected components. 
\end{definition}
Let $k$ be as in Definition~\ref{def:debrujin}. We can also consider $D$ as the transition graph of a finite 
automaton with set of initial states $\{((u,-k),(v,-k)) : u,v \in A^{2r}\}$ 
and set of final states $\{((u,k+1),(v,k+1)) : u,v\in A^{2r}\}$.
Denote by $L(D)$ the language recognized by this finite automaton.

\begin{lemma}\label{prop:inj-dec-p-1}
Injectivity is decidable for p$\nca$ with structural period $p=1$.
\end{lemma}
\begin{proof}
Let $H$ be a p\nca with period $p=1$ and let $D$ its reduced product graph. We prove that $H$ is injective if and only if $L(D)\subseteq 0^*$.

If $L(D)\not\subseteq 0^*$, there exists a word 
$w\in L(D)$ such that $w_i=1$, for some $i$. By definition of $D$, this means that there are at least two distinct configurations which have the same image by $H$. Hence $H$ is not injective.

If $H$ is not injective, then there are two distinct configurations
$x$ and $y$ such that $H(x)=H(y)$. Let $i\in\z$ be such that
$x_i\neq y_i$ and set $m = \max(|i|,k+1)$. For any $j\in\z$, define $u^j = x_{[j-r,j+r)}$ et $v^j= y_{[j-r,j+r)}$. Then,
the path on $D$
\begin{multline*}
(((u^{-m},-k),(v^{-m},-k)),\dots,((u^{-k},-k),(v^{-k},-k)),\dots,((u^0,0),(v^0,0)),\dots,\\
((u^{k+1},k+1),(v^{k+1},k+1)),\dots,((u^m,k+1),(v^m,k+1)))
\end{multline*}
starts from an initial state, ends at a final state, and  contains an edge labelled with $1$. Hence, $L(D)\not\subseteq 0^*$.
\qed
\end{proof}
\begin{theorem}
Surjectivity and injectivity are decidable for p$\nca$.
\end{theorem}
\begin{proof}
Let $H$ be a p\nca of radius $r$. If $p=1$, by Lemma~\ref{lem:surj-dec-p-1} (resp., Lemma~\ref{prop:inj-dec-p-1}) we can decide surjectivity (resp., injectivity) of $H$. Otherwise, without loss of generality, assume $p=r$. By Proposition~\ref{prop:conjug}, $H$ is topologically conjugated to a p\nca $H'$ with structural period 1. By Lemma~\ref{lem:surj-dec-p-1} (resp., Lemma~\ref{prop:inj-dec-p-1}) and since $H$ is surjective (resp., injective) iff $H'$ is as well, we can decide surjectivity (resp., injectivity) of $H$.
\qed
\end{proof}
\subsection{Injectivity  and surjectivity: structural implications}
We now study how informations (about surjectivity or
injectivity) on the global rule $H$ of a p$\nca$ with structural period 1 relate to
properties of the composing local rules.

\begin{proposition}\label{prop:injsurj}
Let $F$ and $G$ be two CA of rules $f$ and $g$, respectively. 
For any p$\nca$ $H$ with period $p=1$ having $f$ and $g$ as right and left default rules, it holds that
\begin{enumerate}
\item $H$ surjective $\Rightarrow F$ surjective and $G$ surjective
\item $H$ injective $\Rightarrow F$ surjective and $G$ surjective
\item $H$ injective on finite configurations $\Rightarrow F$
surjective and $G$ surjective
\end{enumerate}
\end{proposition}
\begin{proof}
Let $k$ be the largest natural such that $h_k\neq f$ or $h_{-k}\neq g$.
Without loss of generality, assume that $F$ is not surjective.
\begin{enumerate}
\item There exists a block $u$ which has no pre-image by $f$.  Let $y$ be any configuration belonging to 
$[u]_{k+1}$.
By definition of $H$, there is no configuration $x\in\az$ with $H(x)=y$.
\item By a theorem 
in~\cite{hedlund69},  $f$ admits a diamond, \ie, there exist $u,v,w\in A^+$ with $u\neq v$ of same length such that $f(wuw)=f(wvw)$. 
Build $x\in [wuw]_{k+1}$ and $y\in [wvw]_{k+1}$ such that $x_i=y_i$ for all $i$ different from the cylinder positions. 
By definition of $H$, $H(x)=H(y)$.
\item the proof is similar to item 2.
\end{enumerate}
\qed
\end{proof}
For d$\nca$ a stronger result holds.
\begin{proposition}
Let $F$ be a CA of local rule $f$. 
For any d$\nca$ $H$ with default rule $f$, it holds that
\begin{enumerate}
\item $H$ injective $\Rightarrow F$ injective 
\item $H$ injective $\Rightarrow H$ surjective
\end{enumerate}
\end{proposition}
\begin{proof}
Let $k$ be the largest natural such that $h_k\neq f$ or $h_{-k}\neq f$.
Fix a configuration $y\in\az$ and for any $u\in A^{2k+1}$ let $y^u\in [u]_{-k}$ be the configuration such that $y^u_i=y_i$ for all $i\in\z$, $|i|>k$. Define $Y=\{y^u: u\in A^{2k+1}\}$ and $X=F^{-1}(Y)$. 
\begin{enumerate}
\item If $H$ is injective then $|X|=|H(X)|$ and by Proposition \ref{prop:injsurj} 
$F$ is surjective. So $X\geq |A|^{2k+1}$. 
By definition of $H$, it holds that $H(X)\subseteq Y$. Hence, 
\[
|A|^{2k+1} \leq |X| = |H(X)| \leq |Y| = |A|^{2k+1}
\]
which gives $|X|=|A|^{2k+1}$ Thus, $F$ is injective.
\item if $H$ is injective we also have $H(X)=Y$. 
Thus $y\in Y$ has a pre-image by $H$.
\end{enumerate}
\qed
\end{proof}

\section{Dynamics}
%
In order to study equicontinuity and almost equicontinuity, we introduce an intermediate class between d\nca and r\nca.
\begin{definition}[$n$-compatible r\nca]
A r\nca $H$ is \emph{$n$-compatible with a local rule} $f$ if for
any $k\in\n$, there exist two integers $k_1
> k$ and $k_2 < -k$ such that
$\forall i \in [k_1,k_1+n) \cup [k_2,k_2+n), \; h_i=f$.
\end{definition}
In other words, a \nca is $n$-compatible with $f$ if, arbitrarily
far from the center of the lattice, there are intervals of length
$n$
in which the local rule $f$ is applied.

The notion of blocking word and the related results cannot be
directly restated in the context of \nca because some words are
blocking just thanks to the uniformity of CA. To overcome this
problem we introduce the following notion.
\begin{definition}[Strongly blocking word]
A word $u\in A^l$ is said to be \emph{strongly $s$-blocking} ($0<s\leq l$) for a
CA $F$ of local rule $f$ if there exists an offset $d \in
[0,l-s]$ such that for any \nca $H$ with $\forall i \in
[0,l)$, $h_i=f$ it holds that
\[
\forall x, y \in [u]_0, \forall n \geq 0,\quad H^n(x)_{[d,d+s)} =
H^n(y)_{[d,d+s)}\enspace.
\]
\end{definition}
Roughly speaking, a word is strongly blocking if it is blocking
whatever be the perturbations involving the rules in its
neighborhood. The following extends Proposition $5.12$
in~\cite{kurka04} to strongly $r$-blocking words.

\begin{theorem}\label{prop:kequi}
Let $F$ be a CA of local rule $f$ and radius $r$. The following statements are equivalent:
\begin{enumerate} 
 \item[(1)] $F$ is equicontinuous;
 \item[(2)] there exists $k>0$ such that any word $u \in A^k$ is strongly $r$-blocking for $F$;
 \item[(3)] any d\nca $H$ of default rule $f$ is ultimately periodic.
\end{enumerate}
\end{theorem}
\begin{proof}
(1) $\Rightarrow$ (2). Suppose that $F$ is equicontinuous. By~\cite[Th. 4]{kurka97}, there exist $p
> 0$ and $q \in \mathbb{N}$ such that $F^{q+p} = F^q$. As a consequence,
we have that
\[
\forall u \in A^*, |u| > 2(q+p)r \Rightarrow f^{p+q}(u) =
f^q(u)_{[pr,|u|-(2q+p)r)}
\enspace.
\]
Let $H$ be a \nca such that $h_j=f$ for each $j\in[0,(2p+2q+1)r)$. For any $x\in\az$ and $i\in\n$,
consider the following words: 
\ignore{
$s^{(i)}= H^i(x)_{[0,qr)}$,
$t^{(i)}= H^i(x)_{[qr,(q+p)r)}$, $u^{(i)}=
H^i(x)_{[(q+p)r,(q+p+1)r)}$, $v^{(i)}=
H^i(x)_{[(q+p+1)r,(q+2p+1)r)}$, $w^{(i)}=
H^i(x)_{[(q+2p+1)r,(2q+2p+1)r)}$.
}
\begin{align*}
s^{(i)}= & H^i(x)_{[0,qr)} \\
t^{(i)}= & H^i(x)_{[qr,(q+p)r)} \\
u^{(i)}= & H^i(x)_{[(q+p)r,(q+p+1)r)} \\
v^{(i)}= & H^i(x)_{[(q+p+1)r,(q+2p+1)r)} \\
w^{(i)}= & H^i(x)_{[(q+2p+1)r,(2q+2p+1)r)}\enspace.
\end{align*}
For all $i \in[0,q+p]$, $u^{(i)}$ is completely
determined by $s^{(0)}t^{(0)}u^{(0)}v^{(0)}w^{(0)} =
x_{[0,(2q+2p+1)r)}$ (see Figure~\ref{motsfortementbloquants}). Moreover, for any natural $i$, we have
\ignore{
$u^{(i+q+p)}= f^{q+p}(s^{(i)}t^{(i)}u^{(i)}v^{(i)}w^{(i)})$ $=
f^q(s^{(i)}t^{(i)}u^{(i)}v^{(i)}w^{(i)})_{[pr,(p+1)r)}
 =  (t^{(i+q)}u^{(i+q)}v^{(i+q)})_{[pr,(p+1)r)} =  u^{(i+q)}$.
 }
\begin{align*}
u^{(i+q+p)} & =  f^{q+p}(s^{(i)}t^{(i)}u^{(i)}v^{(i)}w^{(i)})\\
&=  f^q(s^{(i)}t^{(i)}u^{(i)}v^{(i)}w^{(i)})_{[pr,(p+1)r)} \\
 &=  (t^{(i+q)}u^{(i+q)}v^{(i+q)})_{[pr,(p+1)r)} \\
 &=  u^{(i+q)}\enspace.
\end{align*}

\begin{center}
\begin{figure}[ht]
 \scalebox{0.8}{\ifx\JPicScale\undefined\def\JPicScale{1}\fi \unitlength
\JPicScale mm
\begin{picture}(147.5,55)(0,0) \linethickness{0.3mm}
\put(25,50){\line(1,0){25}} \put(25,45){\line(0,1){5}}
\put(50,45){\line(0,1){5}} \put(25,45){\line(1,0){25}}
\linethickness{0.3mm} \put(105,50){\line(1,0){25}}
\put(105,45){\line(0,1){5}} \put(130,45){\line(0,1){5}}
\put(105,45){\line(1,0){25}} \linethickness{0.1mm}
\multiput(25,45)(0.12,-0.12){375}{\line(1,0){0.12}}
\linethickness{0.1mm}
\multiput(85,0)(0.12,0.12){375}{\line(1,0){0.12}}
\linethickness{0.3mm} \put(70,5){\line(1,0){15}}
\put(70,0){\line(0,1){5}} \put(85,0){\line(0,1){5}}
\put(70,0){\line(1,0){15}} \linethickness{0.1mm}
\multiput(50,45)(0.12,-0.12){167}{\line(1,0){0.12}}
\linethickness{0.1mm}
\multiput(85,25)(0.12,0.12){167}{\line(1,0){0.12}}
\put(77.5,47.5){\makebox(0,0)[cc]{$u^{(i)}$}}

\put(60,47.5){\makebox(0,0)[cc]{$t^{(i)}$}}

\put(37.5,47.5){\makebox(0,0)[cc]{$s^{(i)}$}}

\put(95,47.5){\makebox(0,0)[cc]{$v^{(i)}$}}

\put(117.5,47.5){\makebox(0,0)[cc]{$w^{(i)}$}}

\put(77.5,22.5){\makebox(0,0)[cc]{$u^{(i+q)}$}}

\put(60,22.5){\makebox(0,0)[cc]{$t^{(i+q)}$}}

\put(95,22.5){\makebox(0,0)[cc]{$v^{(i+q)}$}}

\put(77.5,2.5){\makebox(0,0)[cc]{$u^{(i+q)}$}}

\linethickness{0.3mm}
\put(25.62,52.5){\line(1,0){23.75}}
\put(49.37,52.5){\vector(1,0){0.12}}
\put(25.62,52.5){\vector(-1,0){0.12}}
\linethickness{0.3mm}
\put(50.62,52.5){\line(1,0){18.76}}
\put(69.38,52.5){\vector(1,0){0.12}}
\put(50.62,52.5){\vector(-1,0){0.12}}
\linethickness{0.3mm}
\put(70.62,52.5){\line(1,0){13.76}}
\put(84.38,52.5){\vector(1,0){0.12}}
\put(70.62,52.5){\vector(-1,0){0.12}}
\linethickness{0.3mm}
\put(85.62,52.5){\line(1,0){18.76}}
\put(104.38,52.5){\vector(1,0){0.12}}
\put(85.62,52.5){\vector(-1,0){0.12}}
\linethickness{0.3mm}
\put(105.63,52.5){\line(1,0){23.75}}
\put(129.38,52.5){\vector(1,0){0.12}}
\put(105.63,52.5){\vector(-1,0){0.12}}
\put(37.5,55){\makebox(0,0)[cc]{$qr$}}

\put(60,55){\makebox(0,0)[cc]{$pr$}}

\put(77.5,55){\makebox(0,0)[cc]{$r$}}

\put(95,55){\makebox(0,0)[cc]{$pr$}}

\put(117.5,55){\makebox(0,0)[cc]{$qr$}}

\linethickness{0.05mm}
\multiput(85,25)(0,1.9){11}{\line(0,1){0.95}}
\linethickness{0.05mm}
\multiput(70,25)(0,1.9){11}{\line(0,1){0.95}}
\linethickness{0.05mm}
\multiput(50,25)(0,1.9){11}{\line(0,1){0.95}}
\linethickness{0.05mm}
\multiput(105,25)(0,1.9){11}{\line(0,1){0.95}}
\linethickness{0.05mm}
\multiput(70,5)(0,2){8}{\line(0,1){1}}
\linethickness{0.05mm}
\multiput(85,5)(0,2){8}{\line(0,1){1}}
\linethickness{0.05mm}
\put(135,20.62){\line(0,1){23.76}}
\put(135,44.38){\vector(0,1){0.12}}
\put(135,20.62){\vector(0,-1){0.12}}
\linethickness{0.05mm}
\put(135,0.62){\line(0,1){18.76}}
\put(135,19.38){\vector(0,1){0.12}}
\put(135,0.62){\vector(0,-1){0.12}}
\put(147.5,32.5){\makebox(0,0)[cc]{$q$ iterations}}

\put(147.5,10){\makebox(0,0)[cc]{$p$ iterations}}

\linethickness{0.3mm}
\put(50,25){\line(1,0){20}}
\put(50,20){\line(0,1){5}}
\put(70,20){\line(0,1){5}}
\put(50,20){\line(1,0){20}}
\linethickness{0.3mm}
\put(70,25){\line(1,0){15}}
\put(70,20){\line(0,1){5}}
\put(85,20){\line(0,1){5}}
\put(70,20){\line(1,0){15}}
\linethickness{0.3mm}
\put(85,25){\line(1,0){20}}
\put(85,20){\line(0,1){5}}
\put(105,20){\line(0,1){5}}
\put(85,20){\line(1,0){20}}
\linethickness{0.3mm}
\put(50,50){\line(1,0){20}}
\put(50,45){\line(0,1){5}}
\put(70,45){\line(0,1){5}}
\put(50,45){\line(1,0){20}}
\linethickness{0.3mm}
\put(70,50){\line(1,0){15}}
\put(70,45){\line(0,1){5}}
\put(85,45){\line(0,1){5}}
\put(70,45){\line(1,0){15}}
\linethickness{0.3mm}
\put(85,50){\line(1,0){20}}
\put(85,45){\line(0,1){5}}
\put(105,45){\line(0,1){5}}
\put(85,45){\line(1,0){20}}
\end{picture}}
 \caption{A strongly blocking word.}
 \label{motsfortementbloquants}
\end{figure}
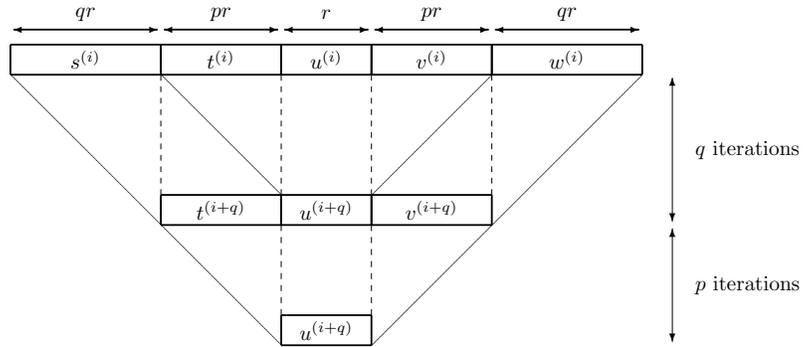
\end{center}
Summarizing, for all $i\in\n$, $u^{(i)}$ is determined by the
word $x_{[0,(2q+2p+1)r)}$ which is then strongly $r$-blocking.
Since $x$ had been chosen arbitrarily, $(2)$ is true.\\
(2) $\Rightarrow$ (3). Since any word of length $k$ is strongly $r$-blocking, any word $u\in  A^{2k+1}$ of length $2k+1$ is ($2r+1$)-blocking, \ie, roughly speaking, $u$ blocks the column $(H^t(z)_{[d,d+2r]})_{t\in\n}$ which appears under itself inside the dynamical evolution of $H$ from any configuration $z\in[u]_0$. 

For any $u\in A^{2k+1}$, consider the configuration $y={}^\infty u^\infty \in [u]_{-k}$ (bi-infinite concatenation of $u$). There exist $q_u$ and $p_u$ such that $F^{p_u + q_u}(y) = F^{q_u}(y)$. 
Set 
$$
q = \max \{ q_u : u \in A^{2k+1} \} \text{ and } p=lcm \{ p_u : u \in A^{2k+1} \}
$$
For any word $u\in  A^{2k+1}$, the column blocked by $u$ admits $q$ and $p$ as pre-period and period, respectively.

Let $H$ be a d\nca of default rule $f$ and let $n$ be such that $\forall i, |i| > n,  h_i = f$.
Choose $x \in A^{\mathbb{Z}}$ and $i \in \mathbb{Z}$ such that  $|i| > n + k$. So $h_{i-k} = \dots = h_{i+k} = f$ and $x_{[i-k, i+k]}$ is a strongly blocking word, then $H^{p+q}(x)_i = H^q(x)_i$.
On the other hand, the sequence $(H^j(x)_{[-n-k,n+k]})_{j \in \mathbb{N}}$ is completely determined by $u = x_{[-m,m]}$ where $m = n+2k+\max \{r_i : i \in \mathbb{Z}\}$ (see Figure~\ref{motsfortementbloquantsbis}).
\begin{center}
\begin{figure}[ht]
 \scalebox{0.8}{\includegraphics{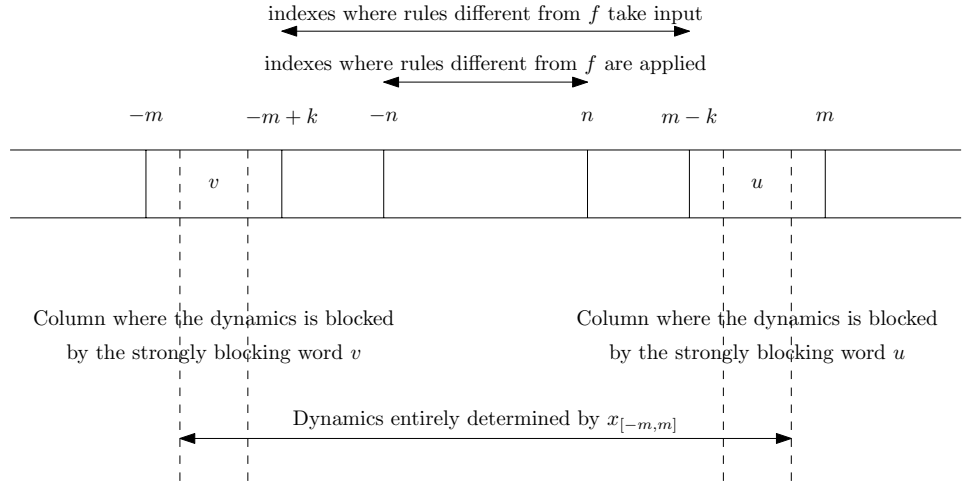}}
 \caption{Dynamics of a d\nca in presence of strongly blocking words.}
 \label{motsfortementbloquantsbis}
\end{figure}
\end{center}
Moreover, there exist $\alpha_u > 0$ and $\beta_u \geq q$ such that
$$
H^{\beta_u}(x)_{[-n-k,n+k]} = H^{\beta_u + p\alpha_u}(x)_{[-n-k,n+k]}
$$
leading to $H^{\beta_u}(x) = H^{\beta_u + p\alpha_u}(x)$. Set now
$$q' = \max \{ \beta_u : u \in A^{2m+1} \} \text{ and } p'=lcm \{ p\alpha_u : u \in A^{2m+1} \}.$$ 
Hence, $\forall x \in A^{\mathbb{Z}}, H^{q'+p'}(x) = H^{q'}(x)$ and $H$ is ultimately periodic.
\\
$(3) \Rightarrow (1)$ Since $H$ is an ultimately periodic d$\nca$ of default rule $f$, the CA $F$ is ultimately periodic too. By~\cite[Th. 4]{kurka97}, $F$ is equicontinuous.
\qed
\end{proof}

\begin{theorem}\label{th:compalmosteq}
Let $F$ be a CA with local rule $f$ admitting a strongly
$r$-blocking word $u$. Let $H$ be a r\nca of radius $r$.  If $H$
is $|u|$-compatible with $f$ then $H$ is almost equicontinuous.
\end{theorem}
\begin{proof}
Let $p$ and $n$ be the offset and the length of $u$, respectively.
For any $k\in\n$, consider the set $T_{u,k}$ of configurations
$x\in\az$ having the following property $\mathcal{P}$: there exist
$l>k$ and $m<-k$ such that $x_{[l,l + n)}=x_{[m,m + n)}= u$ and
$\forall i \in [l,l+n) \cup [m,m+n) \; h_i=f$.
\ignore{\begin{equation} \label{eq:tuk} x_{[l,l + n)}=x_{[m,m +
n)}= u \quad \wedge\quad \forall i \in [l,l+n) \cup [m,m+n) \;
h_i=f\enspace.
\end{equation}
} Remark that $T_{u,k}$ is open, being a union of cylinders.
Clearly, each $T_{u,k}$ is dense, thus the set $T_u = \bigcap_{k
\geq n} T_{u,k}$ is residual. We claim that any configuration in
$T_u$ is an equicontinuity point. Indeed, choose arbitrarily $x
\in T_u$. Set $\epsilon = 2^{-k}$, where $k\in\n$ is such that
$x\in T_{u,k}$. Then, there exist $k_1> k$ and $k_2 <
-k-n$ satisfying $\mathcal{P}$ (see Figure~\ref{pointequicontinuite}).\\
\begin{figure}[ht]
\centering \scalebox{0.7}{\ifx\JPicScale\undefined\def\JPicScale{1}\fi
\unitlength \JPicScale mm
\begin{picture}(140,57.5)(0,0)
\linethickness{0.3mm}
\put(50,40){\line(1,0){40}}
\put(50,35){\line(0,1){5}}
\put(90,35){\line(0,1){5}}
\put(50,35){\line(1,0){40}}
\linethickness{0.3mm}
\put(65,42.5){\line(0,1){5}}
\put(65,42.5){\vector(0,-1){0.12}}
\linethickness{0.3mm}
\put(75,42.5){\line(0,1){5}}
\put(75,42.5){\vector(0,-1){0.12}}
\linethickness{0.3mm}
\put(0,40){\line(1,0){50}}
\put(0,35){\line(0,1){5}}
\put(50,35){\line(0,1){5}}
\put(0,35){\line(1,0){50}}
\put(25,37.5){\makebox(0,0)[cc]{$u$}}

\linethickness{0.3mm}
\put(50,42.5){\line(0,1){5}}
\put(50,42.5){\vector(0,-1){0.12}}
\linethickness{0.3mm}
\put(90,42.5){\line(0,1){5}}
\put(90,42.5){\vector(0,-1){0.12}}
\linethickness{0.3mm}
\put(15,42.5){\line(0,1){5}}
\put(15,42.5){\vector(0,-1){0.12}}
\linethickness{0.3mm}
\put(35,42.5){\line(0,1){5}}
\put(35,42.5){\vector(0,-1){0.12}}
\linethickness{0.3mm}
\put(-0,42.5){\line(0,1){5}}
\put(-0,42.5){\vector(0,-1){0.12}}
\put(-0,50){\makebox(0,0)[cc]{$k_2$}}

\put(15,50){\makebox(0,0)[cc]{$k_2 + p$}}

\put(35,50){\makebox(0,0)[cc]{$k_2 + p + r$}}

\put(50,50){\makebox(0,0)[cc]{$k_2 + |u|$}}

\put(65,50){\makebox(0,0)[cc]{$-k$}}

\put(75,50){\makebox(0,0)[cc]{$k$}}

\put(90,50){\makebox(0,0)[cc]{$k_1$}}

\linethickness{0.3mm}
\put(0,55){\line(1,0){50}}
\put(50,55){\vector(1,0){0.12}}
\put(0,55){\vector(-1,0){0.12}}
\put(25,57.5){\makebox(0,0)[cc]{$f$}}

\linethickness{0.3mm}
\put(15,5){\line(0,1){30}}
\linethickness{0.3mm}
\put(35,5){\line(0,1){30}}
\linethickness{0.3mm}
\multiput(25,0)(0,1.82){6}{\line(0,1){0.91}}
\linethickness{0.3mm}
\put(90,40){\line(1,0){50}}
\put(90,35){\line(0,1){5}}
\put(140,35){\line(0,1){5}}
\put(90,35){\line(1,0){50}}
\put(115,37.5){\makebox(0,0)[cc]{$u$}}

\linethickness{0.3mm}
\put(105,42.5){\line(0,1){5}}
\put(105,42.5){\vector(0,-1){0.12}}
\linethickness{0.3mm}
\put(125,42.5){\line(0,1){5}}
\put(125,42.5){\vector(0,-1){0.12}}
\linethickness{0.3mm}
\put(140,42.5){\line(0,1){5}}
\put(140,42.5){\vector(0,-1){0.12}}
\put(105,50){\makebox(0,0)[cc]{$k_1 + p$}}

\put(125,50){\makebox(0,0)[cc]{$k_1 + p +r$}}

\put(140,50){\makebox(0,0)[cc]{$k_1 + |u|$}}

\linethickness{0.3mm}
\put(90,55){\line(1,0){50}}
\put(140,55){\vector(1,0){0.12}}
\put(90,55){\vector(-1,0){0.12}}
\put(115,57.5){\makebox(0,0)[cc]{$f$}}

\linethickness{0.3mm}
\put(105,5){\line(0,1){30}}
\linethickness{0.3mm}
\put(125,5){\line(0,1){30}}
\linethickness{0.3mm}
\multiput(115,0)(0,1.82){6}{\line(0,1){0.91}}
\end{picture}} \caption{An
equicontinuity point (see Theorem~\protect\ref{th:compalmosteq}).}
\label{pointequicontinuite}
\end{figure}
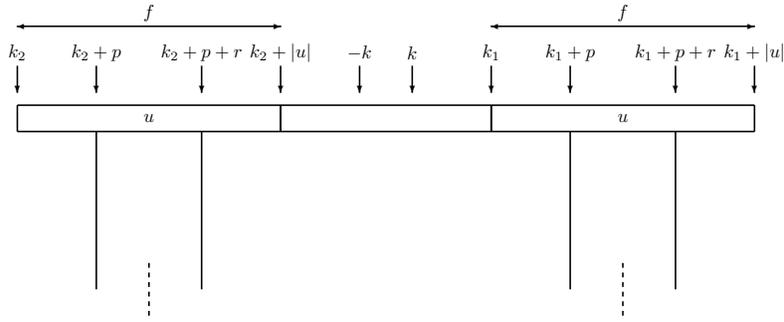
\\
Fix $\delta = \min \{2^{-(k_1+n)},2^{-k_2}\}$ and let $y\in\az$
be such that $d(x,y) < \delta$. Then $y_{[k_2,k_1 + |u|)} =
x_{[k_2,k_1 + |u|)}$. Since $u$ is $r$-blocking, $\forall t\in\n$,
$H^t(x)$ and $H^t(y)$ are equal inside the intervals
$[k_1+p,k_1+p+r]$
 and $[k_2+p,k_2+p+r]$,
then $d(H^t(x),H^t(y)) < \epsilon$.\qed
\end{proof}

In a similar manner one can prove the following.

\begin{theorem}
Let $F$ be an equicontinuous CA  of local rule $f$. Let $k\in\n$
be as in Proposition~\ref{prop:kequi}. Any r\nca $k$-compatible
with $f$ is equicontinuous.
\end{theorem}
\subsection{Perturbing almost equicontinuous CA}
In the sequel, we show how the loss of uniformity may lead to a dramatic change in the dynamical behavior of the automata network.

\begin{example}[An almost equicontinuous CA]
\label{exa:caeq}
Let $A = \{0,1,2\}$ and define a local rule $f: A^3\to A$ as
follows: $\forall x, y \in A$,\[
\begin{array}{c}
f(x,0,y)=
  \begin{cases}
  1 &\text{if $x = 1$ or $y = 1$}\\
  0 &\text{otherwise}
  \end{cases}
\\
f(x,1,y) =
  \begin{cases}
  2 &\text{if $x = 2$ or $y = 2$}\\
  1 &\text{otherwise}
  \end{cases}
\\
f(x,2,y) =
  \begin{cases}
  0 &\text{if $x = 1$ or $y = 1$}\\
  2 &\text{otherwise}\enspace.
  \end{cases}
\end{array}
\]
 
\end{example}
We show that the CA defined in Example~\ref{exa:caeq} is almost equicontinuous.
\ignore{
\begin{proposition}
The CA defined in Example~\ref{exa:caeq} is almost equicontinuous.
\end{proposition}
}

\begin{proof}
Just remark that the number of 0s inside the word $20^i2$ is
non-decreasing. Thus $202$ is
a $1$-blocking word (see Figure~\ref{motifpreserve}).
\qed
\end{proof}

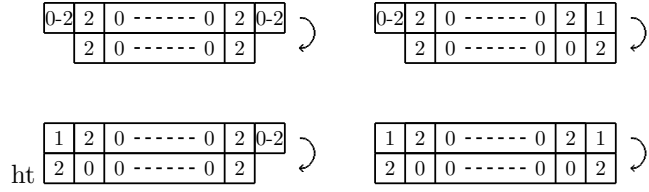
\begin{figure}{ht}
\centering \scalebox{0.8}{\ifx\JPicScale\undefined\def\JPicScale{1}\fi
\unitlength \JPicScale mm
\begin{picture}(100,30)(0,0)
\linethickness{0.3mm}
\put(5,10){\line(1,0){30}}
\put(5,5){\line(0,1){5}}
\put(35,5){\line(0,1){5}}
\put(5,5){\line(1,0){30}}
\linethickness{0.3mm}
\put(10,5){\line(0,1){5}}
\linethickness{0.3mm}
\put(30,5){\line(0,1){5}}
\put(7.5,7.5){\makebox(0,0)[cc]{2}}

\put(32.5,7.5){\makebox(0,0)[cc]{2}}

\put(12.5,7.5){\makebox(0,0)[cc]{0}}

\put(27.5,7.5){\makebox(0,0)[cc]{0}}

\linethickness{0.3mm}
\multiput(15,7.5)(1.82,0){6}{\line(1,0){0.91}}
\linethickness{0.3mm}
\put(5,30){\line(1,0){30}}
\put(5,25){\line(0,1){5}}
\put(35,25){\line(0,1){5}}
\put(5,25){\line(1,0){30}}
\linethickness{0.3mm}
\put(10,25){\line(0,1){5}}
\linethickness{0.3mm}
\put(30,25){\line(0,1){5}}
\put(7.5,27.5){\makebox(0,0)[cc]{2}}

\put(32.5,27.5){\makebox(0,0)[cc]{2}}

\put(12.5,27.5){\makebox(0,0)[cc]{0}}

\put(27.5,27.5){\makebox(0,0)[cc]{0}}

\linethickness{0.3mm}
\multiput(15,27.5)(1.82,0){6}{\line(1,0){0.91}}
\linethickness{0.3mm}
\put(5,25){\line(1,0){30}}
\put(5,20){\line(0,1){5}}
\put(35,20){\line(0,1){5}}
\put(5,20){\line(1,0){30}}
\linethickness{0.3mm}
\put(10,20){\line(0,1){5}}
\linethickness{0.3mm}
\put(30,20){\line(0,1){5}}
\put(7.5,22.5){\makebox(0,0)[cc]{2}}

\put(32.5,22.5){\makebox(0,0)[cc]{2}}

\put(12.5,22.5){\makebox(0,0)[cc]{0}}

\put(27.5,22.5){\makebox(0,0)[cc]{0}}

\linethickness{0.3mm}
\multiput(15,22.5)(1.82,0){6}{\line(1,0){0.91}}
\linethickness{0.3mm}
\put(0,5){\line(1,0){35}}
\put(0,0){\line(0,1){5}}
\put(35,0){\line(0,1){5}}
\put(0,0){\line(1,0){35}}
\linethickness{0.3mm}
\put(10,0){\line(0,1){5}}
\linethickness{0.3mm}
\put(30,0){\line(0,1){5}}
\put(2.5,2.5){\makebox(0,0)[cc]{2}}

\put(32.5,2.5){\makebox(0,0)[cc]{2}}

\put(12.5,2.5){\makebox(0,0)[cc]{0}}

\put(27.5,2.5){\makebox(0,0)[cc]{0}}

\linethickness{0.3mm}
\multiput(15,2.5)(1.82,0){6}{\line(1,0){0.91}}
\put(7.5,2.5){\makebox(0,0)[cc]{0}}

\linethickness{0.3mm}
\put(5,0){\line(0,1){5}}
\linethickness{0.3mm}
\multiput(42.5,22.5)(0.49,0.05){1}{\line(1,0){0.49}}
\multiput(42.99,22.55)(0.47,0.14){1}{\line(1,0){0.47}}
\multiput(43.46,22.69)(0.22,0.12){2}{\line(1,0){0.22}}
\multiput(43.89,22.92)(0.13,0.1){3}{\line(1,0){0.13}}
\multiput(44.27,23.23)(0.1,0.13){3}{\line(0,1){0.13}}
\multiput(44.58,23.61)(0.12,0.22){2}{\line(0,1){0.22}}
\multiput(44.81,24.04)(0.14,0.47){1}{\line(0,1){0.47}}
\multiput(44.95,24.51)(0.05,0.49){1}{\line(0,1){0.49}}
\multiput(44.95,25.49)(0.05,-0.49){1}{\line(0,-1){0.49}}
\multiput(44.81,25.96)(0.14,-0.47){1}{\line(0,-1){0.47}}
\multiput(44.58,26.39)(0.12,-0.22){2}{\line(0,-1){0.22}}
\multiput(44.27,26.77)(0.1,-0.13){3}{\line(0,-1){0.13}}
\multiput(43.89,27.08)(0.13,-0.1){3}{\line(1,0){0.13}}
\multiput(43.46,27.31)(0.22,-0.12){2}{\line(1,0){0.22}}
\multiput(42.99,27.45)(0.47,-0.14){1}{\line(1,0){0.47}}
\multiput(42.5,27.5)(0.49,-0.05){1}{\line(1,0){0.49}}

\linethickness{0.3mm}
\multiput(42.5,22.5)(0.12,0.12){5}{\line(0,1){0.12}}
\linethickness{0.3mm}
\multiput(42.5,22.5)(0.12,-0.12){5}{\line(0,-1){0.12}}
\linethickness{0.3mm}
\multiput(42.5,2.5)(0.49,0.05){1}{\line(1,0){0.49}}
\multiput(42.99,2.55)(0.47,0.14){1}{\line(1,0){0.47}}
\multiput(43.46,2.69)(0.22,0.12){2}{\line(1,0){0.22}}
\multiput(43.89,2.92)(0.13,0.1){3}{\line(1,0){0.13}}
\multiput(44.27,3.23)(0.1,0.13){3}{\line(0,1){0.13}}
\multiput(44.58,3.61)(0.12,0.22){2}{\line(0,1){0.22}}
\multiput(44.81,4.04)(0.14,0.47){1}{\line(0,1){0.47}}
\multiput(44.95,4.51)(0.05,0.49){1}{\line(0,1){0.49}}
\multiput(44.95,5.49)(0.05,-0.49){1}{\line(0,-1){0.49}}
\multiput(44.81,5.96)(0.14,-0.47){1}{\line(0,-1){0.47}}
\multiput(44.58,6.39)(0.12,-0.22){2}{\line(0,-1){0.22}}
\multiput(44.27,6.77)(0.1,-0.13){3}{\line(0,-1){0.13}}
\multiput(43.89,7.08)(0.13,-0.1){3}{\line(1,0){0.13}}
\multiput(43.46,7.31)(0.22,-0.12){2}{\line(1,0){0.22}}
\multiput(42.99,7.45)(0.47,-0.14){1}{\line(1,0){0.47}}
\multiput(42.5,7.5)(0.49,-0.05){1}{\line(1,0){0.49}}

\linethickness{0.3mm}
\multiput(42.5,2.5)(0.12,0.12){5}{\line(0,1){0.12}}
\linethickness{0.3mm}
\multiput(42.5,2.5)(0.12,-0.12){5}{\line(0,-1){0.12}}
\linethickness{0.3mm}
\put(35,30){\line(1,0){5}}
\put(35,25){\line(0,1){5}}
\put(40,25){\line(0,1){5}}
\put(35,25){\line(1,0){5}}
\linethickness{0.3mm}
\put(0,30){\line(1,0){5}}
\put(0,25){\line(0,1){5}}
\put(5,25){\line(0,1){5}}
\put(0,25){\line(1,0){5}}
\linethickness{0.3mm}
\put(0,10){\line(1,0){5}}
\put(0,5){\line(0,1){5}}
\put(5,5){\line(0,1){5}}
\put(0,5){\line(1,0){5}}
\linethickness{0.3mm}
\put(35,10){\line(1,0){5}}
\put(35,5){\line(0,1){5}}
\put(40,5){\line(0,1){5}}
\put(35,5){\line(1,0){5}}
\linethickness{0.3mm}
\put(60,30){\line(1,0){30}}
\put(60,25){\line(0,1){5}}
\put(90,25){\line(0,1){5}}
\put(60,25){\line(1,0){30}}
\linethickness{0.3mm}
\put(65,25){\line(0,1){5}}
\linethickness{0.3mm}
\put(85,25){\line(0,1){5}}
\put(62.5,27.5){\makebox(0,0)[cc]{2}}

\put(87.5,27.5){\makebox(0,0)[cc]{2}}

\put(67.5,27.5){\makebox(0,0)[cc]{0}}

\put(82.5,27.5){\makebox(0,0)[cc]{0}}

\linethickness{0.3mm}
\multiput(70,27.5)(1.82,0){6}{\line(1,0){0.91}}
\linethickness{0.3mm}
\put(60,25){\line(1,0){35}}
\put(60,20){\line(0,1){5}}
\put(95,20){\line(0,1){5}}
\put(60,20){\line(1,0){35}}
\linethickness{0.3mm}
\put(65,20){\line(0,1){5}}
\linethickness{0.3mm}
\put(85,20){\line(0,1){5}}
\put(62.5,22.5){\makebox(0,0)[cc]{2}}

\put(92.5,22.5){\makebox(0,0)[cc]{2}}

\put(67.5,22.5){\makebox(0,0)[cc]{0}}

\put(82.5,22.5){\makebox(0,0)[cc]{0}}

\linethickness{0.3mm}
\multiput(70,22.5)(1.82,0){6}{\line(1,0){0.91}}
\put(87.5,22.5){\makebox(0,0)[cc]{0}}

\linethickness{0.3mm}
\put(90,20){\line(0,1){5}}
\linethickness{0.3mm}
\multiput(97.5,22.5)(0.49,0.05){1}{\line(1,0){0.49}}
\multiput(97.99,22.55)(0.47,0.14){1}{\line(1,0){0.47}}
\multiput(98.46,22.69)(0.22,0.12){2}{\line(1,0){0.22}}
\multiput(98.89,22.92)(0.13,0.1){3}{\line(1,0){0.13}}
\multiput(99.27,23.23)(0.1,0.13){3}{\line(0,1){0.13}}
\multiput(99.58,23.61)(0.12,0.22){2}{\line(0,1){0.22}}
\multiput(99.81,24.04)(0.14,0.47){1}{\line(0,1){0.47}}
\multiput(99.95,24.51)(0.05,0.49){1}{\line(0,1){0.49}}
\multiput(99.95,25.49)(0.05,-0.49){1}{\line(0,-1){0.49}}
\multiput(99.81,25.96)(0.14,-0.47){1}{\line(0,-1){0.47}}
\multiput(99.58,26.39)(0.12,-0.22){2}{\line(0,-1){0.22}}
\multiput(99.27,26.77)(0.1,-0.13){3}{\line(0,-1){0.13}}
\multiput(98.89,27.08)(0.13,-0.1){3}{\line(1,0){0.13}}
\multiput(98.46,27.31)(0.22,-0.12){2}{\line(1,0){0.22}}
\multiput(97.99,27.45)(0.47,-0.14){1}{\line(1,0){0.47}}
\multiput(97.5,27.5)(0.49,-0.05){1}{\line(1,0){0.49}}

\linethickness{0.3mm}
\multiput(97.5,22.5)(0.12,0.12){5}{\line(1,0){0.12}}
\linethickness{0.3mm}
\multiput(97.5,22.5)(0.12,-0.12){5}{\line(1,0){0.12}}
\linethickness{0.3mm}
\put(55,30){\line(1,0){5}}
\put(55,25){\line(0,1){5}}
\put(60,25){\line(0,1){5}}
\put(55,25){\line(1,0){5}}
\linethickness{0.3mm}
\put(90,30){\line(1,0){5}}
\put(90,25){\line(0,1){5}}
\put(95,25){\line(0,1){5}}
\put(90,25){\line(1,0){5}}
\linethickness{0.3mm}
\put(60,10){\line(1,0){30}}
\put(60,5){\line(0,1){5}}
\put(90,5){\line(0,1){5}}
\put(60,5){\line(1,0){30}}
\linethickness{0.3mm}
\put(65,5){\line(0,1){5}}
\linethickness{0.3mm}
\put(85,5){\line(0,1){5}}
\put(62.5,7.5){\makebox(0,0)[cc]{2}}

\put(87.5,7.5){\makebox(0,0)[cc]{2}}

\put(67.5,7.5){\makebox(0,0)[cc]{0}}

\put(82.5,7.5){\makebox(0,0)[cc]{0}}

\linethickness{0.3mm}
\multiput(70,7.5)(1.82,0){6}{\line(1,0){0.91}}
\linethickness{0.3mm}
\put(55,5){\line(1,0){40}}
\put(55,0){\line(0,1){5}}
\put(95,0){\line(0,1){5}}
\put(55,0){\line(1,0){40}}
\linethickness{0.3mm}
\put(65,0){\line(0,1){5}}
\linethickness{0.3mm}
\put(85,0){\line(0,1){5}}
\put(57.5,2.5){\makebox(0,0)[cc]{2}}

\put(92.5,2.5){\makebox(0,0)[cc]{2}}

\put(67.5,2.5){\makebox(0,0)[cc]{0}}

\put(82.5,2.5){\makebox(0,0)[cc]{0}}

\linethickness{0.3mm}
\multiput(70,2.5)(1.82,0){6}{\line(1,0){0.91}}
\linethickness{0.3mm}
\put(90,0){\line(0,1){5}}
\linethickness{0.3mm}
\put(60,0){\line(0,1){5}}
\put(62.5,2.5){\makebox(0,0)[cc]{0}}

\put(87.5,2.5){\makebox(0,0)[cc]{0}}

\linethickness{0.3mm}
\multiput(97.5,2.5)(0.49,0.05){1}{\line(1,0){0.49}}
\multiput(97.99,2.55)(0.47,0.14){1}{\line(1,0){0.47}}
\multiput(98.46,2.69)(0.22,0.12){2}{\line(1,0){0.22}}
\multiput(98.89,2.92)(0.13,0.1){3}{\line(1,0){0.13}}
\multiput(99.27,3.23)(0.1,0.13){3}{\line(0,1){0.13}}
\multiput(99.58,3.61)(0.12,0.22){2}{\line(0,1){0.22}}
\multiput(99.81,4.04)(0.14,0.47){1}{\line(0,1){0.47}}
\multiput(99.95,4.51)(0.05,0.49){1}{\line(0,1){0.49}}
\multiput(99.95,5.49)(0.05,-0.49){1}{\line(0,-1){0.49}}
\multiput(99.81,5.96)(0.14,-0.47){1}{\line(0,-1){0.47}}
\multiput(99.58,6.39)(0.12,-0.22){2}{\line(0,-1){0.22}}
\multiput(99.27,6.77)(0.1,-0.13){3}{\line(0,-1){0.13}}
\multiput(98.89,7.08)(0.13,-0.1){3}{\line(1,0){0.13}}
\multiput(98.46,7.31)(0.22,-0.12){2}{\line(1,0){0.22}}
\multiput(97.99,7.45)(0.47,-0.14){1}{\line(1,0){0.47}}
\multiput(97.5,7.5)(0.49,-0.05){1}{\line(1,0){0.49}}

\linethickness{0.3mm}
\multiput(97.5,2.5)(0.12,0.12){5}{\line(1,0){0.12}}
\linethickness{0.3mm}
\multiput(97.5,2.5)(0.12,-0.12){5}{\line(1,0){0.12}}
\linethickness{0.3mm}
\put(55,10){\line(1,0){40}}
\put(55,5){\line(0,1){5}}
\put(95,5){\line(0,1){5}}
\put(55,5){\line(1,0){40}}
\put(2.5,7.5){\makebox(0,0)[cc]{1}}

\put(2.5,27.5){\makebox(0,0)[cc]{0-2}}

\put(92.5,27.5){\makebox(0,0)[cc]{1}}

\put(57.5,7.5){\makebox(0,0)[cc]{1}}

\put(92.5,7.5){\makebox(0,0)[cc]{1}}

\put(57.5,27.5){\makebox(0,0)[cc]{0-2}}

\put(37.5,27.5){\makebox(0,0)[cc]{0-2}}

\put(37.5,7.5){\makebox(0,0)[cc]{0-2}}

\end{picture}} \caption{Evolution of
words $20^i2$ according to $F^{(9)}$.} \label{motifpreserve}
\end{figure}

The following example defines a \nca which is sensitive to the
initial conditions although its default rule give rise to an
almost equicontinuous CA.
\begin{example}[A sensitive \nca with an almost equicontinuous default rule]\label{ex:aeqsens}\mbox{}\\
Consider the d\nca $\H:\az\to\az$ defined as follows
\[
\forall x \in\az, \forall i \in\z,\quad \H(x)_i =
\begin{cases}
1 & \text{if $i=0$} \\
f(x_{i-1},x_i,x_{i+1}) & \text{otherwise\enspace,}
\end{cases}
\]
where $f$ and $A$ are as in Example~\ref{exa:caeq}.
\end{example}
\begin{figure}
\begin{center}
\subfloat[][Dynamics of {$F^{(9)}$}]{\includegraphics[scale=0.5]{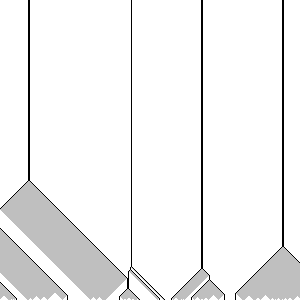}}\qquad
\subfloat[][Dynamics of {$\H$}]{\includegraphics[scale=0.5]{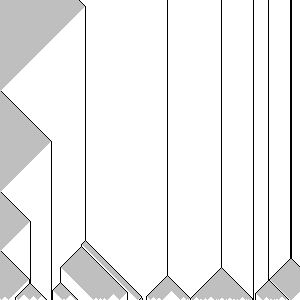}}\\
\subfloat[][Dynamics of {$\H$} on $u0^\infty$]{
\includegraphics[scale=0.5]{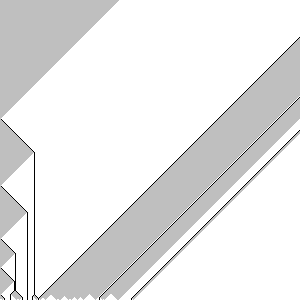}}\qquad
\subfloat[][Dynamics of {$\H$} on $u2^\infty$]{
\includegraphics[scale=0.5]{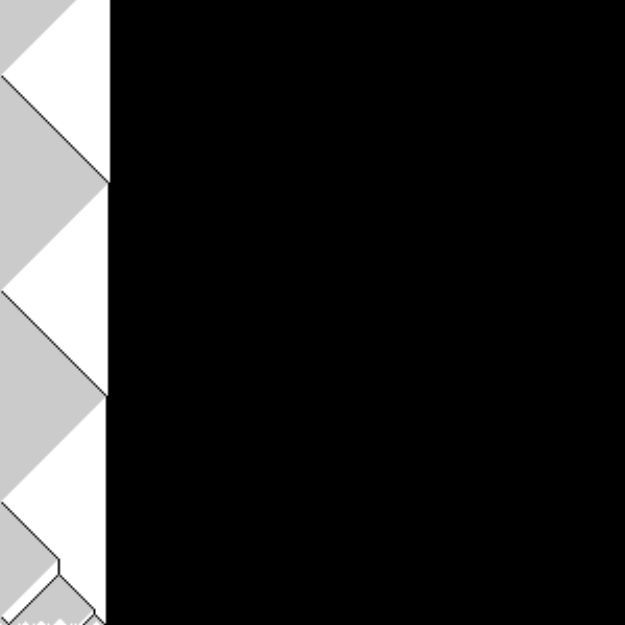}}
\caption{Space-time diagrams for $F^{(9)}$ and $\H$.}
\end{center}
\end{figure}
Remark that positive and negative cells do not interact each other under the action of $\H$. Therefore, in order to study the
behavior of \H, it is sufficient to consider the action of \H
on \an. In the sequel, we will simply note by $H$ the map \H.

In order to prove that \H is sensitive, we need some technical Lemmata.
\begin{lemma}
\label{regle3}
For any $u \in A^*$, consider the sequence $(u^{(n)})_{n \in \mathbb{N}}$ defined as:
\[
\begin{cases}
u^{(n+1)} = f(1u^{(n)}0) & \forall n\in\n \\
u^{(0)} = u2 & 
\end{cases}
\]
Then,
\[
\exists m,\, \forall n\geq m, \quad u^{(n)} = 1^{|u|+1}
\]
\ignore{
$$\lim_{n \rightarrow \infty} u^{(n)} = 1^{|u|+1}$$
}
\end{lemma}

\begin{proof}
We proceed in 4 steps.
\begin{enumerate}
\item
First of all, we are going to show that there exists $n_0\in\n$ such that $u^{(n_0)}$ does not contain any $1$. In particular, we prove that there exists $n_0\in\n$ such that $\forall n\leq n_0$ the integer
\[
i^{(n)} = \min \left\{ i\leq |u| : u^{(n)}_{i} = 2 \text{ and } u^{(n)}_{[i + 1, |u|]} \in \{0,2\}^*\right\}
\]
is well defined with the property $P(n) = \left( \forall k\in[0, n), i^{(k + 1)} \leq i^{(k)}\right)$ and $i^{(n_0)} = 0$. By definition, $i^{(0)}$ is well defined and the property $P(0)$ is true. Suppose now that, for some $n\in\n$, $i^{(n)}$ is well defined and the property $P(n)$ is true. We deal with the following cases.
\begin{enumerate}
\item
If $i^{(n)} = 0$, then we set $n_0 = n$ and we are done.
\item
If $i^{(n)} \neq 0$ does not contain any $1$, then we can write $u^{(n)} = 0^{i^{(n)}}2w$ with $w\in\{0,2\}^*$ and we have $u^{(n + 1)} = 10^{i^{(n)} - 1}2w$. So, $i^{(n+1)} = i^{(n)}$ is well defined with $P(n+1)$ true and, having the element $n+1$ as a starting point, we fall in the next case.
\item
If $i^{(n)} \neq 0$ contains at least one $1$, let $h\in [0,i^{(n)})$ be the greatest position in which it appears and set $s=i^{(n)}-h-1$. We can write $u^{(n)} =v^{(n)}10^{s}2w$, for some $v^{(n)} \in A^{h}$ and $w\in\{0,2\}^*$. Then,  for each $j\in[1,s]$, we obtain $u^{(n + j)} = v^{(n + j)}10^{s - j}2w$, for some $v^{(n + j)} \in A^{h + j}$. So, for each $j\in[1,s]$,  $i^{(n + j)} = i^{(n)}$ is well defined with $P(n+j)$ true. Furthermore, it holds that $u^{(n+s+1)} = v^{(n+s+1)}20w$, for some $v^{(n+s+1)} \in A^{h+s+1}$, and 
$i^{(n+s+1)} = i^{(n+s+1)} - 1$ is well defined with $P(n+s+1)$ true, and we can reconsider the three cases having the element $n+s+1$ as starting point. 
\end{enumerate}
By considering iteratively the three cases, we are sure to reach a natural  $n_0$ such that $i^{(n_0)}=0$ since whenever we fall in the third case the value of $i^{(n)}$ decreases.  
\item
Proceeding by induction, we now show that
\[
 \forall n \geq n_0, \exists k \in \mathbb{N}, \exists v \in \{0,2\}^*,\quad\text{s.t.}\quad u^{(n)} = 1^{k}v.
 \] 
Clearly, this is true for $n=n_0$ with $k=0$. Assume now that the statement is true for some $n\geq n_0$ and consider the following cases. 
\begin{itemize}
\item
If $v = \epsilon$, then $u^{(n+1)} = u^{(n)} = 1^{k}v$
\item
If $v_0 = 0$, then $u^{(n+1)} = 1^{k + 1}v_{[1,|v|-1]}$
\item
If $v_0 = 2$ and $k \neq 0$, then $u^{(n+1)} = 1^{k - 1}20v_{[1,|v|-1]}$
\item
If $v_0 = 2$ and $k = 0$, then $u^{(n+1)} = 0v_{[1,|v|-1]}$
\end{itemize}
In all the cases, the statement is true for $n+1$. 

As a consequence, we also have that the number $|u^{(n)}|_2 $ of $2$ inside $u^{(n)}$ is a (non strictly) decreasing sequence:
\[
 \forall n \geq n_0, \quad |u^{(n+1)}|_2 \leq |u^{(n)}|_2
\]
Indeed,  $u^{(n)}$ does not contain the block $121$, which, transforming itself into $202$, is the unique one able to increase the number of $2$.
\item
We now prove that there exists $n_1\geq n_0$, such that $u^{(n_1)}$ no longer contains any $2$, and then $u^{(n_1)}=1\cdots 10\cdots 0$. This is assured by showing that
\[
\forall n \geq n_0, |u^{(n)}|_2 > 0 \Rightarrow \exists s \in \mathbb{N}, |u^{(n + s)}|_2 < |u^{(n)}|_2
\]
Let $n \geq n_0$ such that $|u^{(n)}|_2 > 0$. Since $u^{(n)} = 1^{k}v$ for some $k\in\n, v \in \{0,2\}^*$, we can write $u^{(n)} = 1^k0^{h}2w$ for some $h\in\n, w \in \{0,2\}^*$. Thus,  we have $u^{(n + h)} = 1^{k+h}2w$ and $u^{(n + h + i)} = 1^{k+h-i}20^iw$ for each $i\in [1,h+k]$. So, $u^{(n + 2h + k)} = 20^{h+k}w$ and, setting $s=2h+k+1$, we obtain $u^{(n + s)} = 0^{s-h}w$, assuring that $|u^{(n + s)}|_2 < |u^{(n)}|_2$.
\item
Since $u^{(n_1)}=1^k0^h$ for some $h,k\in\n$, it is easy to observe that $u^{(n_1+i)}=1^{k+i}0^{h-i}$, $i\in[1,h]$. In particular, setting $m=n_1+h$, we obtain $u^{(m)}=1^{|u|+1}$ and $\forall n\geq m, u^{(n)}=u^{(m)}$. This concludes the proof. 
\end{enumerate}
\end{proof}
\begin{lemma}
\label{rewriting}
Consider the rewriting system on $A^* \times \{0,1\}$ defined by the following rules:
\begin{enumerate}
\item
$(u0,0)  \xrightarrow{1}  (u,0)$
\item
$(u1,0)  \xrightarrow{2}  (u,1)$
\item
$(u2,0)  \xrightarrow{3}  (f(1u20),0)$
\item
$(u0,1)  \xrightarrow{4}  (f(1u0),1)$
\item
$(u1,1)  \xrightarrow{5}  (u,1)$
\item
$(u2,1)  \xrightarrow{6}  (f(1u2),0)$
\item
$(\epsilon, x)  \xrightarrow{7}  (\epsilon, 1)$
\end{enumerate}
Starting from any $(u,x)\in A^* \times \{0,1\}$, after a certain number $m$ of rule applications, the system ultimately falls into $(\epsilon, 1)$.
\end{lemma}

\begin{proof}
This system is non ambiguous and then, for any $(u,x)\in A^* \times \{0,1\}$,  it (well) defines the sequence $(u^{(n)}, x^{(n)})_ {n \in \mathbb{N}}$ such that 
\[
\begin{cases}
(u^{(n)}, x^{(n)}) \rightarrow (u^{(n+1)}, x^{(n+1)}) & \forall n\in\n \\
(u^{(0)},  x^{(0)}) = (u,x) & 
\end{cases}
\]
where $\rightarrow$ is the unique possible application of a system rule.
Consider the sequence $(l^{(n)})_ {n \in \mathbb{N}}=(|u^{(n)}|)_ {n \in \mathbb{N}}$. By definition, it is a (non strictly) decreasing sequence and then it converges to some $l\in\n$, or, equivalently, there exists $m\in\n$ such that $\forall n\geq m$, $l^{(n)}=l$. We show that $l=0$ and this also prove the thesis.  For a sake of argument, suppose that $l > 0$. Thus, there exists $k\in\n$ such that $\forall n \geq k, (u^{(n)}, x^{(n)}) \xrightarrow{3} (u^{(n+1)}, x^{(n+1)})$ since, except rule 7, rule 3 is the only one leaving $l^{(n)}$ unchanged. Furthermore, the sequence $(u^{(n+k)})_{n \in \mathbb{N}}$ verifies the hypothesis of Lemma \ref{regle3} and so it is ultimately equal to $1^l$, that is contrary to the fact that $\forall n \geq k, u^{(n)}_l = 2$, since rule 3 is always applied 
\end{proof}

\begin{lemma}
\label{interdits}
Let $\mathcal{F} = \{01, 12, 20, 22\}$. For any $x \in A^{\mathbb{N}}$ and any $i\in\n$, if no element of $\mathcal{F}$ appears inside  $x_{[i, \infty)}$, then no element of $\mathcal{F}$ appears inside $\H(x)_{[i+1, \infty)}$.
\end{lemma}

\begin{proof}
The $f$--pre-images of words in $\mathcal{F}$ are :
\begin{itemize}
\item
$f^{-1}(01)$ = \{0001, 1201, 2001\}
\item
$f^{-1}(12)$ = \{0012, 0112, 1012, 1020, 1022, 1112, 2012\}
\item
$f^{-1}(20)$ = \{0120, 0121, 0122, 0200, 0202, 0221, 1120, 1121, 1122, 2120, 2121, 2122, 2200, 2202, 2221\}
\item
$f^{-1}(22)$ = \{0220, 0222, 2112, 2220, 2222\}
\end{itemize}
So, if there exists $w\in\mathcal{F}$ appearing in $\H(x)_{[i+1, \infty)}$, necessarily a word $u\in\mathcal{F}$ is inside $x_{[i, \infty)}$.
\end{proof}
\begin{lemma}
\label{sens1}
For any $u \in A^{*}$, there exists $n_0\in\n$ s.t. $\forall n > n_0, H^n(u0^{\infty})_1 = 1$.
\end{lemma}
\begin{proof}
Consider the sequences $(u^{(n)}, x^{(n)})_{n \in \mathbb{N}}$ and  $(l^{(n)})_ {n \in \mathbb{N}}$ from Lemma~\ref{rewriting} in which $(u^{(0)}, x^{(0)}) = (u_{[1,|u|-1]},0)$. Define the sequences $(k^{(n)})_{n \in \mathbb{N}}$ and $(y^{(n)})_{n \in \mathbb{N}}$ as follows:
\[
\begin{cases}
k^{(n+1)} = k^{(n)} + 1& \forall n\in\n \;\text{s.t.}\; (u^{(n)}, x^{(n)}) \xrightarrow{a} (u^{(n+1)}, x^{(n+1)}) \quad a=3,4,6,7  \\
k^{(n+1)} = k^{(n)} & \forall n\in\n\;\text{s.t.}\;(u^{(n)}, x^{(n)}) \xrightarrow{a} (u^{(n+1)}, x^{(n+1)})\quad a=1,2,5\\
k^{(0)} = 0& 
\end{cases},
\]
and, $\forall n\in\n$, $y^{n}=H^{k^{(n)}}(u0^{\infty})$, respectively.  
First of all, we are going to prove that the property 
\[
L(n)= \left( y^{n}_{[1,l^{(n)} + 1]} = u^{(n)}x^{(n)}\right),
\]
linking the dynamics of $H$ with the one induced by the rewriting system, and the property 
\[
M(n)= \left( \forall w\in\mathcal{F},\quad\text{$w$ does not appears inside}\; y^{n}_{[l^{(n)} + 1, \infty)} \right), 
\] 
are valid for all $n\in\n$. 
We proceed by induction.

The properties $L(n)$ and $M(n)$ statement are clearly true for $n=0$. Suppose now that they are valid for some $n\in\mathbb{N}$ and let the rewriting system evolve on $(u^{(n)}, x^{(n)})$ according to the rule $a$, for some $a=1,\ldots, 7$. 

If $a\in\{1,2,5\}$,  then $k^{(n+1)} = k^{(n)}$ and $l^{(n+1)} = l^{(n)} - 1$.  Hence, $L(n+1)$ is true. Moreover, when the restriction passes from $[l^{(n)} + 1, \infty)$ to $[l^{(n+1)} + 1, \infty)$, the additional word inside the configuration $y^{n+1}=y^{n}$ is either $00$, or $10$, or $11$, depending on the value of $a$. So, $M(n+1)$ is also valid.

If $a=7$, then $k^{(n+1)} = k^{(n)}+1$ and $l^{(n+1)} = l^{(n)}=0$.
Since $L(n)$ and $M(n)$ are true, $y^{n}_{[0,2]} \in \{ 100, 102, 110, 111 \}$, and then $x^{(n+1)} = 1 =y^{n+1}_1$. So, $L(n+1)$ is valid. Moreover, no element $\mathcal{F}$ appears inside $y^{n}_{[1, \infty)}$ and neither inside $y^{n}_{[0, \infty)}$, and, hence, by Lemma~\ref{interdits}, neither inside $y^{n+1}_{[1, \infty)}$. Thus, $M(n+1)$ is true.

If $a\in\{3,4,6\}$, then $k^{(n+1)} = k^{(n)} + 1$. By the fact that $L(n)$ is true, we have $u^{(n+1)} = f(y^{n}_{[0, l^{(n+1)} + 1]})$, or, equivalently, $u^{(n+1)} = y^{n+1}_{[1, l^{(n+1)}]}$. Since $M(n)$ is true, by Lemma~\ref{interdits}, $y^{n+1}_{[l^{(n)} + 2, \infty)}$ does not contain any element of $\mathcal{F}$. It remains to prove that $x^{(n+1)} = y^{n+1}_{ l^{(n+1)} + 1}$, and there is no word of $\mathcal{F}$ inside $y^{n+1}_{[ l^{(n+1)}+1, l^{(n)}+2]}$. 

If $a=3$, $l^{(n+1)}=l^{(n)}$ and we have  $y^{n+1}_{ l^{(n+1)} + 1} = f(20a)$ where necessarily $a\neq 1$ since $M(n)$ is valid and $01\in\mathcal{F}$. Hence,  $y^{n+1}_{ l^{(n+1)} + 1} = 0 = x^{(n+1)}$ and $L(n+1)$ is true. For a sake of argument, assume that the word $w=y^{n+1}_{[ l^{(n+1)}+1, l^{(n)}+2]}\in\mathcal{F}$. Necessarily, we obtain $w=01$ and so, by definition of $f$, $y^{n}_{[ l^{(n)}+2, l^{(n)}+3]} = 01$, that is a contradiction. Hence, $M(n+1)$ is valid.

If $a=4$ (resp., $a=6$), $l^{(n+1)}=l^{(n)} -1$ and we have  $y^{n+1}_{ l^{(n+1)} + 1} = f(a01)$ (resp., $f(a21)$). So,  $y^{n+1}_{ l^{(n+1)} + 1} = 0 = x^{(n+1)}$ and $L(n+1)$ is true. Since $M(n)$ is true, $y^{n}_{[ l^{(n)}, l^{(n)}+3]} = 01bc$ (resp., $21bc$) with  $bc\in\{00,02,10,11\}$. By definition of $f$, it follows that $y^{n+1}_{[ l^{(n+1)}, l^{(n+1)}+2]} = 111$ and then $M(n+1)$ is valid.
 
Summarizing, we have proved that $L(n)$ and $M(n)$ are valid for all $n\in\n$ and, in particular, 
\[
\forall n \in \mathbb{N}, \quad H^{k^{(n)}}(u0^{\infty})_{[1,l^{(n)} + 1]} = u^{(n)}x^{(n)}
\]
Now, let $m$ be the integer from Lemma~\ref{rewriting}. Recall that $\forall n\geq m$, $l^{(n)} = 0$ and $(u^{(n)}, x^{(n)}) = (\epsilon, 1)$. Thus, setting $n_0=k^{(m)}$, we obtain
\[
\forall n\geq n_0,\quad H^{n}(u0^{\infty})_1 = 1\xspace .
\]
\end{proof}

\begin{lemma}
\label{sens2}
For any $u\in A^{*}$ and any $n\in\n$, there exists $m > n$ such that $H^m(u2^{\infty})_1 = 2$.
\end{lemma}

\begin{proof}
For any $u \in A^{*}$ and any $n\in\n$, define the configuration $z^{n} = H^n(u2^{\infty})$ and the integers 
$a^{(n)} = \max \{i \in \mathbb{N} : {z_i^{n}} = 1\}$ and $b^{(n)} = \min \{i \in \mathbb{N} : {z_i ^{n}}= 2 \wedge \forall j > i, {z_j^{n}}\neq 1\}$. Remark that $a^{(n)}$ and $b^{(n)}$ are well defined and $\forall n\in\n$, $a^{(n)} < b^{(n)}$. We want to prove that $\forall n\in\n, (b^{(n)} > 1 \Rightarrow \exists k > n, b^{(k)} < b^{(n)})$.
Let $n\in\n$ be such that $b^{(n)} > 1$. Since $z^{n}_{[a^{(n)}, b^{(n)}]} = 10^{b^{(n)} - a^{(n)} - 1}2$ and $\forall i > b^{(n)}$, ${z_i^{n}}\neq 1$, by definition of $f$, it follows that $a^{(n+i)} = a^{(n)} + i$, $b^{(n+i)} = b^{(n)}$, for $i\in[0,b^{(n)} - a^{(n)} - 1]$. Hence, $b^{(n +  b^{(n)} - a^{(n)})} =  b^{(n)} - 1<b^{(n)}$.
\end{proof}
We conclude stating that \H is sensitive.
\begin{proof}
It is a direct consequence of Lemmata~\ref{sens1} and~ \ref{sens2}.
\end{proof}
The following example shows that default rules individually
defining almost equicontinuous CA can also constitute \nca that
have a completely different behavior from the one in
Example~\ref{ex:aeqsens}.

\begin{example}[An equicontinuous \nca made by almost equicontinuous CA]\mbox{}\\
Let $A = \{0,1,2\}$ and define the local rule $f: A^3\to A$ as:
$\forall x, y, z\in A$, $f(x,y,z)=2$ if $x=2$ or $y=2$ or $z=2$,
$z$ otherwise. 
\ignore{\[ f(x,y,z) =
\begin{cases}
2 &\text{if $x=2$ or $y=2$ or $z=2$} \\
z &\text{otherwise}\enspace.
\end{cases}
\]
} The CA $F$ of local rule $f$ is almost equicontinuous since $2$
is a blocking word. The restriction of $F$ to $\{0,1\}^{\z}$ gives
the shift map which is sensitive. Thus $F$ is not equicontinuous.
Define now the following d\nca \H:
\[
\forall x\in\az,\forall i\in\z,\quad \H(x)_i =
\begin{cases}
2 &\text{if $i=0$} \\
f(x_{i-1},x_i,x_{i+1}) &\text{otherwise}\enspace.
\end{cases}
\]
\end{example}
\ignore{
\begin{proposition}
The d\nca \H is equicontinuous.
\end{proposition}
}
We now prove that $\H$ is equicontinuous.
\begin{proof}
Let $n\in\n,x,y \in\az$ be such that $x_{[-2n,2n]} =
y_{[-2n,2n]}$. Since $H$ is of radius $1$, $\forall k \leq n,
H^k(x)_{[-n,n]} = H^k(y)_{[-n,n]}$ and $\forall k > n$,
$H^k(x)_{[-n,n]} = 2^{2n+1} = H^k(y)_{[-n,n]}$. So, $H$ is
equicontinuous.
\end{proof}
\section{Conclusions}
This paper have introduced \nca, an extension of CA model obtained by relaxing the
uniformity property (\ie, the fact that the same local rule is applied
to all sites of the CA lattice). The study of how this change affects
the dynamics of the systems has just started. We proved several
results concerning basic set properties like injectivity and surjectivity.
Moreover, we studied how informations about the \nca can determine
properties on the underlying CA or vice-versa.

The study of \nca  can be continued along
several different directions. Of course, it would be interesting to progress in the analysis of the dynamical behavior. In particular we believe it would be worthwhile to study how information
moves along the space-time diagrams and how the density of changes
affects the entropy of the system.

It is well-known that CA cannot be used a random generator and,
in general, they are poor (but fast) random generators. Can
\nca be a better tool in this context?
\section{Acknowledgements}
This work has been supported by the French ANR Blanc Projet ``EMC'', by the ``Coop\'eration
scientifique international R\'egion Provence--Alpes--C\^ote
d'Azur'' Project ``Automates Cellulaires, dynamique symbolique et
d\'ecidabilit\'e'', and by the PRIN/MIUR project ``Mathematical
aspects and forthcoming applications of automata and formal
languages''.
\bibliographystyle{plain}
\bibliography{latajou}
\end{document}